\documentclass[american,aps,prl,reprint]{revtex4-2}
\usepackage[T1]{fontenc}
\usepackage[latin9]{inputenc}
\usepackage{color}
\usepackage{babel}
\usepackage{amsmath}
\usepackage{amsthm}
\usepackage{amssymb}
\usepackage{graphicx}
\usepackage[pdfusetitle,
 bookmarks=true,bookmarksnumbered=false,bookmarksopen=false,
 breaklinks=false,pdfborder={0 0 0},pdfborderstyle={},backref=false,colorlinks=true]
 {hyperref}
\hypersetup{
 allcolors=magenta}

\makeatletter
\theoremstyle{plain}
\newtheorem{thm}{\protect\theoremname}
\theoremstyle{plain}
\newtheorem{lem}[thm]{\protect\lemmaname}
\theoremstyle{plain}
\newtheorem{prop}[thm]{\protect\propositionname}
\theoremstyle{plain}
\newtheorem{coro}[thm]{\protect\corollaryname}

\usepackage{times}
\usepackage{txfonts}
\usepackage{braket}
\usepackage{colortbl}

\makeatother

\providecommand{\lemmaname}{Lemma}
\providecommand{\propositionname}{Proposition}
\providecommand{\theoremname}{Theorem}
\providecommand{\corollaryname}{Corollary}

\begin{document}
\title{Asymptotic Limits of Entanglement Transmission}

\author{Piotr Masajada}
\author{Aby Philip}
\author{Alexander Streltsov}
\email{streltsov.physics@gmail.com}
\affiliation{Institute of Fundamental Technological Research, Polish Academy of Sciences, Pawi\'{n}skiego 5B, 02-106 Warsaw, Poland}

\begin{abstract}
Reliable distribution of quantum entanglement over long distances is a central challenge in quantum information science, fundamentally limited by decoherence in noisy communication channels. In this work, we investigate the asymptotic limits of entanglement distribution across homogeneous linear repeater chains with arbitrary intermediate LOCC processing. We establish a strict dichotomy: the asymptotic preservation of entanglement over arbitrarily long distances is possible if and only if the underlying quantum channel admits a correctable subspace. For channels lacking such a subspace, we prove that the transmitted state converges exponentially fast to the set of separable states, rendering standard LOCC filtering insufficient. To counteract this exponential degradation, we analyze ``networks'' employing parallel channel uses per link. We derive a fundamental lower bound on the required number of parallel channel uses per elementary link, proving that for broad classes of channels without a correctable subspace, the number of parallel channels per link must scale at least logarithmically with the number of intermediate stations to sustain a non-zero amount of entanglement. This provides a code-independent bound for the number of physical links for the considered homogeneous one-way architecture. For depolarizing noise below the corresponding coding threshold, the logarithmic lower bound is achievable.
\end{abstract}

\maketitle

\emph{\textbf{Introduction.}} Quantum entanglement~\cite{Horodeckientanglementreview} is a fundamental resource for a wide array of quantum information processing tasks, including quantum key distribution~\cite{BennettQKD}, quantum teleportation~\cite{Bennettteleportofstate}, and distributed quantum computing~\cite{Kimblequantuminternet,Meterdq}. A central challenge in realizing these technologies is the reliable distribution of entanglement between distant parties. In realistic scenarios, quantum states must be transmitted through noisy quantum channels, which inevitably interact with the environment and induce decoherence. This continuous degradation strictly limits the amount of entanglement that can be shared and maintained between remote parties. 

In the standard entanglement distribution scenario~\cite{Palnotmest,Streltsovdistribution,Krisnandadistributionentanglement,Zuppardentanglementdistribution} two distant parties, Alice and Bob, have access to a noisy quantum channel \( \Lambda \). Alice first prepares a bipartite quantum state \( \rho \) in her laboratory, keeping one subsystem and using the channel to transmit the other subsystem to Bob. Because the transmission is imperfect, the channel acts nontrivially on the particle sent to Bob, while the particle retained by Alice remains unchanged. As a result, the initially prepared state \( \rho \) is transformed into the final bipartite state \(\Lambda \otimes \openone [\rho] \). Any entanglement initially present in \( \rho \) is generally degraded by the action of the channel, and the properties of \( \Lambda \) determine how much entanglement can ultimately be shared between the two parties.

If the noisy channel \( \Lambda \) is used sequentially, the resulting state shared by Alice and Bob is given by  
\begin{equation}
    (\Lambda \circ \cdots \circ \Lambda) \otimes \openone [\rho].
    \label{eq::lambdaserial}
\end{equation}
 This setting is relevant, for example, in long-distance quantum communication through a chain of noisy transmission segments, where the quantum system effectively undergoes several consecutive applications of the same channel before reaching Bob. In general, the detrimental effects of noise accumulate with the number of channel uses, leading to a progressive degradation of the shared entanglement~\cite{LamiEBindices}. In many physically relevant situations, this entanglement eventually disappears altogether in the limit of infinitely many channel applications~\cite{Rahamaneb}.

To counteract the accumulation of noise, one can introduce intermediate stations to process the transmitted systems between successive channel segments. This general principle underlies quantum repeaters, although different repeater architectures employ different physical resources and classical communication methods~\cite{Briegelrepeater,Muralidharannetwork,Duanrepeaters, Muralidharanlogarithmics,Munrorepeater,Azumarepeaterstation,Azumarepeaters}. Here, we consider a homogeneous chain in which every elementary link is described by the same finite-dimensional channel $\Lambda$. Each transmission round consists of a trace-preserving LOCC operation $\mathcal{F}_i$ followed by an application of $\Lambda\otimes\openone$. The operations may depend on all previous measurement outcomes; equivalently, the classical transcript required for such adaptivity may be included in the state, see Fig.~\ref{fig:Setting}. Denoting the overall transformation after $n$ rounds by $\Lambda_n$, we have
\begin{equation}
    \label{eq:LambdaN}
    \Lambda_n[\rho]
    =
    (\Lambda\otimes\openone)\circ\mathcal{F}_n\circ\cdots\circ
    (\Lambda\otimes\openone)\circ\mathcal{F}_1[\rho].
\end{equation}

For deterministic protocols, our entanglement measure of choice is the
entanglement of formation
\cite{Hillentanglementofformation,Woottersentanglementofformation},
\begin{equation}
    E_{\mathrm f}(\rho^{AB})
    =
    \min_{\{p_i,\psi_i^{AB}\}}
    \sum_i p_i
    S\!\left(\operatorname{Tr}_B\psi_i^{AB}\right),
\end{equation}
where the minimum is taken over all pure-state decompositions
$\rho^{AB}=\sum_i p_i\psi_i^{AB}$ and $S(\rho^{AB})$ denotes von Neumann entropy. For two-qubit states,
$E_{\mathrm f}$ is a strictly increasing function of the concurrence
$C$,
\begin{equation}
    E_{\mathrm f}(\rho)
    =
    h_2\!\left(
        \frac{1+\sqrt{1-C(\rho)^2}}{2}
    \right),
\end{equation}
where $h_2$ is the binary entropy. For the definitions of concurrence see Supplementary Material. Thus, for two qubits, convergence of concurrence
to zero is equivalent to convergence of the entanglement of formation
to zero.

We say that a deterministic protocol
asymptotically preserves entanglement if, for some initial state
$\rho^{AB}$ and some admissible sequence of intermediate operations,
\begin{equation}
    \label{eq:NonzeroLimit}
    \lim_{n\rightarrow\infty}
    E_{\mathrm f}\!\left(\Lambda_n[\rho^{AB}]\right)>0.
\end{equation}
Without intermediate processing, the related problem was studied in
the theory of asymptotically entanglement-saving channels
\cite{Lamieas}. Our setting additionally permits arbitrary adaptive
operations between successive channel uses.

We subsequently consider another transmission architecture in which
each occurrence of $\Lambda$ is replaced by $m$ parallel uses
$\Lambda^{\otimes m}$. The intermediate operations may include
decoding, quantum error correction, recovery, and re-encoding.
Consequently, deterministic error-correction-based
third-generation repeaters form a special case of our model
\cite{Muralidharannetwork}. Concrete realizations of this architecture
include direct one-way protocols based on teleportation-based error
correction, encoded Bell measurements, block quantum codes, and code
concatenation
\cite{Muralidharanlogarithmics,namikirepeater,Ewertrepeater,Muralidharanrepeater,Worepeater,NiuqLDPC,Glaudellrepeater}.
Our analysis does not assume any particular code, decoding procedure,
or recovery scheme.

This problem is complementary to the literature on quantum network
capacities, routing, and cut-based bounds
\cite{Azumanetwork,Rigovaccanetworks,Pirandolanetwork,
Baumlnetwork,Pantnetworks}. These frameworks typically fix a network
and optimize an asymptotic communication rate, allowing general
topologies, adaptive protocols, and different routing strategies.
Instead, we let the number of links $n$ in a homogeneous chain grow
and ask how the number of physical channels $m=m(n)$ assigned to every link must scale
for the end-to-end entanglement to remain nonzero. Within this setting, we derive
code-independent constraints on asymptotic entanglement distribution
and determine when the resulting resource scaling is tight. Techniques like network capacities
and min-cut bounds do not directly determine this finite-block
scaling with the chain length.

\begin{figure}
\includegraphics[width=0.8\columnwidth]{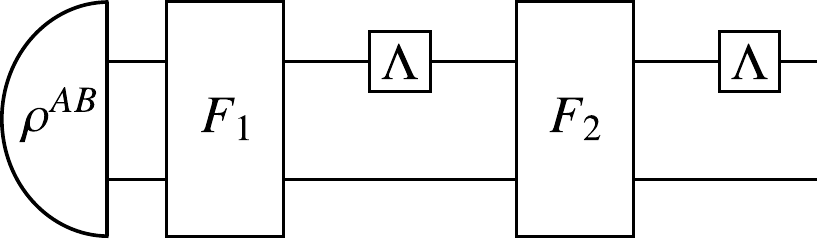}
\caption{\label{fig:Setting}
\textbf{Entanglement distribution through a noisy channel with intermediate repeater stations.} An initial bipartite state \( \rho^{AB} \) is shared between parties \( A \) and \( B \). Subsystem \( A \) is transmitted through a sequence of \( n \) noisy channel uses \( \Lambda \), which may be understood as propagation through \( n \) successive transmission segments with intermediate stations located between Alice and Bob. After each channel use, the parties, possibly assisted by these intermediate stations, may apply an LOCC filter \( F_i \). The figure illustrates the resulting state in Eq.~\eqref{eq:LambdaN} for $n=2$.}
\end{figure}

\medskip

\emph{\textbf{Single-qubit channels.}} To begin, we investigate entanglement distribution through a single-qubit quantum channel $\Lambda$. We assume Alice holds a single-qubit memory that she aims to entangle with Bob's system across a sequence of $n$ transmission segments, each linked via the quantum channel $\Lambda$. We recall that the Choi state of  \(\Lambda\) is defined as~\cite{CHOIch}
\begin{equation}
\Gamma = (\Lambda \otimes \openone)(\phi^{+}),
\end{equation}
where \(\ket{\phi^{+}} = (\ket{00} + \ket{11})/\sqrt{2}\) is a Bell state. 

We consider the scenario in which the intermediate filtering operations $F_i$ are allowed to be stochastic local operations and classical communication (SLOCC)~\cite{Durslocc} rather than deterministic LOCC. This broader framework is important because probabilistic filters can, upon success, increase the entanglement of the post-selected state which is not possible with deterministic LOCC. The key question is thus whether this additional freedom provided by SLOCC can preserve a nonzero amount of entanglement after many sequential applications of the channel by conditioning on successful filtering events. In the following, we show that although SLOCC filters may temporarily improve the quality of the surviving states, this does not overcome the fundamental limitations of the channel: any such apparent entanglement preservation comes at the price of a success probability that necessarily decreases with the number of channel uses. To show this, we use concurrence whose definition we recall in the Supplementary Material. Note that going forward, $P(C \geq c,n)$ denotes the probability that the final state $\Lambda_n[\rho^{AB}]$ has concurrence at least $c$.

\begin{thm}\label{thm::qubit_slocc}
Let $\Lambda$ be a single-qubit channel and consider an adaptive
$n$-round protocol in which an SLOCC instrument is followed by
$\Lambda\otimes\openone$ in each round. The instrument in round $i$
may depend on the previous outcomes
$h_{i-1}=(x_1,\ldots,x_{i-1})$. If $\rho_{H_n}$ denotes the random
final state, then, for every $c>0$,
\begin{equation}
    \Pr\!\left[C(\rho_{H_n})\geq c\right]
    \leq
    \frac{q^n C(\rho_0)}{c}
    \leq
    \frac{q^n}{c},
\end{equation}
where
$q=C[(\Lambda\otimes\openone)(\phi^+)]$ is the concurrence of Choi state of $\Lambda.$
\end{thm}
\noindent The proof of the theorem can be found in the supplemental materials. If we demand that  $P(C\geq c,n)=1$ in Theorem~\ref{thm::qubit_slocc}, we recover the following statement.

\begin{coro}{\label{thm::qubit}}
For any single-qubit channel $\Lambda$ the amount of concurrence
decreases as 
\begin{equation}
C(\Lambda_{n}[\rho^{AB}])\leq q^{n}.
\end{equation}
\end{coro}
\noindent We also note that this corollary has been proven using methods in~\cite{Konradfactorization}. We provide a proof in the supplemental materials.
Concurrence is related with entanglement of formation by the following inequality:
\begin{equation}
    E_f(\rho)\leq C(\rho).
\end{equation}
Thus, if concurrence converges to 0, so does entanglement of formation. Since a two-qubit state has concurrence one only if it is a pure maximally entangled state~\cite{Woottersentanglementofformation}, the corresponding Choi state has rank one, implying that $\Lambda$ is unitary. Thus, only qubit unitary channels lead to preservation of entanglement in the limit of infinite filter--channel sequence. 

\medskip

This kind of problem was previously studied in two important aspects. First, in Ref.~\cite{LamiEBindices}, authors studied under which conditions an infinite sequence of qubit filter--channel becomes an entanglement breaking channel. We demand stronger property, we demand that entanglement of formation after infinite filter--channel sequence does not converge to 0. In Ref.~\cite{Lamieas},
it was shown that, without intermediate operations, an entanglement of a state does not converge to 0 if and only if the channel is unitary.
Corollary~\ref{thm::qubit} strengthens these qualitative results by
showing that, even when intermediate LOCC operations are allowed, every
nonunitary qubit channel imposes an exponentially decaying upper bound
on the surviving concurrence. Consequently, unitary channels are the
only qubit channels for which such exponential decay can be avoided.

\emph{\textbf{General channels.}} We now consider the more general setting in which subsystem \( A \) has an arbitrary finite dimension \( d_A \geq 2 \), and the channel \( \Lambda \) acts on \( \mathcal{H}_A \). This extension goes beyond the single-qubit case and requires more refined tools, as two-qubit entanglement monotones such as concurrence are no longer directly applicable in higher dimensions. In the following, we say that a quantum channel \( \Lambda \) admits a correctable subspace \( \mathcal{H}_c \subseteq \mathcal{H} \) if there exists a recovery channel \( \mathcal{R} \) such that
\begin{equation}
\mathcal{R}\circ\Lambda[\psi]=\psi
\end{equation}
for all \( \ket{\psi} \in \mathcal{H}_c \) \cite{Blumerecovery}. 

Having established the necessary groundwork, we now prove the following theorem.

\begin{thm} \label{thm:General}
For any quantum channel $\Lambda$, there exists a state $\rho^{AB}$ such that
\begin{equation}
\lim_{n\rightarrow\infty}E_\mathrm f(\Lambda_{n}[\rho^{AB}])>0,\label{eq:Theorem2Inequality}
\end{equation}
if and only if $\Lambda$ contains a correctable subspace.
\end{thm}
\noindent We refer to the supplemental material for the proof.

A related problem was studied in Ref.~\cite{Lamieas}, where the authors consider asymptotically entanglement-saving (AES) channels, that preserve a nonzero amount of entanglement after inifinitely many repeated applications of the channel without intermediate operations. Our setting is more general, since we allow intermediate LOCC recovery operations.
To see that this leads to a strictly larger class of channels, consider a channel $\Lambda_{\mathrm{ex}}$ acting on a four-dimensional system, with Kraus operators
\begin{equation}
    K_1=\ket{2}\bra{0}+\ket{3}\bra{1},
    \qquad
    K_2=\ket{0}\bra{2},
    \qquad
    K_3=\ket{0}\bra{3}.
\end{equation}
The subspace
$\mathcal{C}=\operatorname{span}\{\ket{0},\ket{1}\}$
is correctable, since it is mapped isometrically onto
$\mathcal{V}=\operatorname{span}\{\ket{2},\ket{3}\}$.
However,
\begin{equation}
\Lambda_{\mathrm{ex}}^{2}(\rho)
=
\operatorname{Tr}(P_{\mathcal{C}}\rho)\ket{0}\bra{0}
+
\operatorname{Tr}(P_{\mathcal{V}}\rho)\ket{2}\bra{2},
\end{equation}
which is a measure-and-prepare channel and is therefore entanglement
breaking~\cite{HorodeckiEB}. Hence, $\Lambda_{\mathrm{ex}}$ is not AES, although
entanglement can be preserved indefinitely by applying a recovery
operation between successive channel uses. Therefore, the class
characterized in Theorem~\ref{thm:General} strictly contains the AES class.

Theorem~\ref{thm:General} establishes a strict qualitative dichotomy for asymptotic entanglement based entirely on the existence of a correctable subspace. A natural question that follows is how the system behaves when such a subspace is absent. In this regime, the entanglement must vanish, but understanding the rate at which it disappears is crucial for characterizing the channel's noise properties. The following theorem provides a quantitative bound, showing that without a correctable subspace, the output state converges exponentially fast to the set of separable states, denoted by $\mathcal{S}$. In the following, $\Vert X\Vert_1= \mathrm{Tr}\sqrt{X^\dagger X}$ denotes the trace norm of a matrix $X$.
\begin{thm}\label{thm:exponential}
For any quantum channel $\Lambda$ acting on a subsystem $A$ of dimension $d_A$ that possesses no correctable subspace, there exists a constant $\kappa \in (0,1)$ such that for any initial bipartite state $\rho^{AB}$, the trace distance to the set of separable states $\mathcal{S}$ satisfies
\begin{equation}
\min_{\sigma^{AB}\in\mathcal{S}}\Vert\Lambda_{n}[\rho^{AB}]-\sigma^{AB}\Vert_{1}\leq 4\kappa^{n/[2(d_{A}-1)]}(\sqrt{d_{A}}-1).
\end{equation}
\end{thm}
\noindent We refer to the supplemental material for the proof.

This exponential decay has severe implications for long-distance quantum communication and entanglement distribution. If we restrict each connection to a single use of a noisy channel lacking a correctable subspace, the state inevitably becomes separable over large distances. To combat this exponential degradation, we will extend the scenario to parallel uses of channels between the links. By utilizing multiple channels simultaneously, we can attempt to offset the noise accumulation. The next theorem illustrates the physical resource cost of this approach, quantifying exactly how many parallel channel uses are required to successfully establish entanglement across a ``network'' of a given length.

First, we assume that Alice and Bob are connected via linear ``network'' consisting of $n$ sequential links. Second, the connection between any two adjacent nodes is modeled by $m$ independent identical channels. Third, at each node we can apply arbitrary LOCC filter, which aims to improve the entanglement between the parties. Specifically, we can apply decoding, error correction and encoding operations. Consequently, direct one-way idealized quantum repeaters of the 3rd generation (in classification from \cite{Muralidharannetwork}) are included as a special case of our model. We assume that all LOCC operations at the repeater stations are noiseless. Since noisy local operations cannot improve the performance of the protocol, a lower bound on the number of channel uses per link obtained in this idealized setting remains valid for realistic noisy repeaters.

Finally, the operational goal is to distribute a bipartite entangled state across the entire ``network'' such that the endpoint state retains a strictly positive amount of entanglement. Under these conditions, Theorem~\ref{th::scaling} provides a lower bound on the resources required to counter the
exponentially fast decay of entanglement:
maintaining a nonvanishing amount of end-to-end entanglement over
$n$ links requires $m=\Omega(\log n)$ channel uses per elementary
link. In the following theorem, by establishing entanglement between $A$ and $B$, we mean that:
\begin{equation}
    \liminf_{n\rightarrow\infty}
    E_f\!\left(\rho_{n,m(n)}^{AB}\right)
    >0,
\end{equation}
where $\rho_{n,m(n)}^{AB}$ is a state after $n$ segments of a network.

\begin{thm}{\label{th::scaling}}
    Assume that we have a quantum ``network'' consisting of $n$ segments. Assume that we send a bipartite quantum state $\rho_{0}^{A_1\cdots A_m|B}=\rho_0^{A|B}$ through this ``network''. Each segment consists of LOCC filter which can perform arbitrary $A\times B\rightarrow A\times B$ LOCC operation in bipartition $A|B.$ After LOCC filter we have $m$ parallel uses of a channel $\Lambda.$ If 
    \begin{enumerate}
        \item $\Lambda$ can be decomposed as
    \begin{equation}
        \Lambda=p\mathcal{E}+(1-p)\mathcal{M},
        \label{eq::th5decom}
    \end{equation}
    where $\mathcal{E}$ is some EB channel and $\mathcal{M}$ is an arbitrary quantum channel, we can establish entanglement at the output of a ``network'' only if $m$ scales as:
    \begin{equation}
        m=\Omega (\ln n).
    \end{equation}
    \item $\Lambda$ is a qudit amplitude damping channel with $\eta_j>0$, we can establish entanglement at the output of the network only if  $m$ scales as:
    \begin{equation}
        m=\Omega (\ln n).
    \end{equation}
    The Kraus operators of a qudit amplitude damping channel are:
    \begin{equation}
    \begin{split}
        &K_0=\sum_{j=0}^{d-1} \sqrt{q_j}\ket{j}\bra{j}\\
        &K_{i\leftarrow j}=\sqrt{\gamma_{i\leftarrow j}}\ket{i}\bra{j},
    \end{split}
\end{equation}
where $\eta_j =\sum_{i<j}\gamma_{i\leftarrow  j}$ is total decay probability of the $j$-th level and $q_j=1-\eta_j$ is no decay probability. If for every $j>0$ $\eta_j>0$ the channel does not have a correctable subspace. 
    \end{enumerate}
\end{thm}

The complete proof can be found in the supplemental material. We provide a brief proof outline here. We first
bound the contraction parameter $\kappa$ for
$\Lambda^{\otimes m}$. For channels admitting the decomposition
\eqref{eq::th5decom}, we exploit the existence of a rank-one Kraus
representation of the entanglement-breaking component. For the
amplitude-damping channel, we instead use a suitable explicit Kraus
representation. We then apply Theorem~\ref{thm:exponential} and a
uniform continuity bound for the entanglement of formation to relate
the distance from the separable set to the number of physical channels per link $m$ and the
number of links $n$. Solving the resulting inequality yields the
logarithmic lower bound.

A related logarithmic lower bound was obtained for fault-tolerant
qubit circuits under local i.i.d.\ noise
\cite{FawziSpaceOverhead}. They showed that preserving an arbitrary
unknown logical qubit over $n$ sequential noisy stages requires
$m=\Omega(\log n)$ noisy physical qubits per stage. Reliable qubit
transmission entails the preservation of entanglement with a reference
system and is therefore closely related to the task considered here.
This result differs from those presented in this paper in both operational criteria and dimension of the channels involved. Ref.~\cite{FawziSpaceOverhead} requires the effective channel to approximate the identity on every input qubit state, whereas we only require that some input state retain a nonvanishing amount of  entanglement at the very end. Also, their result is restricted to single-qubit noise, while Theorem~\ref{th::scaling} applies, under its stated assumptions, to channels of arbitrary fixed finite dimension.

Every full-Kraus-rank $d$-dimensional channel, i.e., a channel of Choi rank $d^2$, can be written in the form \eqref{eq::th5decom}~\cite{Georgeebnoneb}. However,
not every channel without a correctable subspace admits such a
decomposition. An important example is the amplitude-damping channel,
which is covered separately by Theorem~\ref{th::scaling}. The theorem
therefore applies to broad classes of physically relevant
finite-dimensional channels.

Finally, we show that the logarithmic lower bound is tight for
depolarizing noise below the correctable-error threshold of the
chosen code family. For a family of quantum codes with linear minimum
distance, the failure probability per elementary link decreases
exponentially with the number of physical channels and takes the form
$e^{-\alpha m}$. Hence, over $n$ links, the total
failure probability is bounded by $n e^{-\alpha m}$, so that
$m=\mathcal{O}(\log n)$ is sufficient. Together with
Theorem~\ref{th::scaling}, this establishes the optimal scaling
$m=\Theta(\log n)$ within our model. The
full derivation of our achievability result is provided in the
supplemental material. A related measurement-based
qLDPC repeater protocol was recently developed by Shi \textit{et al.}
\cite{Shirepeater}, who demonstrated constant-yield Bell-state
distribution over arbitrary distances and exponential suppression of
the logical error with increasing code size. Their protocol employs
entanglement distillation and swapping and uses a different resource
accounting; therefore, it does not establish the same minimal number of physical links
statement as Theorem~\ref{th::scaling}, but provides a complementary
qLDPC-based approach to long-distance entanglement distribution.

\medskip

\emph{\textbf{Discussion.}}
Our results form a sequence of structural, quantitative, and
operational limitations on long-distance entanglement distribution.
Theorem~\ref{thm:General} first identifies the exact structural condition under which
intermediate deterministic LOCC can preserve some entanglement
indefinitely: the channel must possess a nontrivial correctable
subspace. Theorem~\ref{thm:exponential} provides a quantitative
statement regarding entanglement preservation in absence of a correctable subspace: the output approaches
the set of separable states exponentially fast in the number of
channel segments, uniformly over the initial state and regardless of the adaptive
intermediate LOCC operations. By applying theorem~\ref{thm:exponential}
to tensor powers of the channel, theorem~\ref{th::scaling} provides a physical resource constraint to counter the exponential
decay. For channels with a
nonzero entanglement-breaking component---including all
full-Kraus-rank channels---and for qudit amplitude-damping channels, theorem~\ref{th::scaling} proves that maintaining
\begin{equation}
    \liminf_{n\rightarrow\infty}
    E_{\mathrm f}\!\left(\rho_{n,m(n)}^{AB}\right)>0
\end{equation}
requires
\begin{equation}
    m=\Omega(\log n).
\end{equation}
Thus, the logarithmic bound is not an isolated network estimate, but
an operational consequence of the underlying entanglement-contraction
theorem.

A logarithmic lower bound with an analogous scaling was previously
obtained for preserving a single unknown logical qubit in
fault-tolerant circuits subject to local i.i.d. qubit noise
\cite{FawziSpaceOverhead}. Theorem~\ref{th::scaling} instead provides an analogous lower bound directly in terms of nonvanishing endpoint entanglement and extends it
to the specified classes of channels of arbitrary fixed finite
dimension. For depolarizing noise in the corresponding correctable
regime, quantum-code families with linear minimum distance achieve
$m=O(\log n)$. Together with theorem~\ref{th::scaling}, this establishes the tight
scaling
\begin{equation}
    m=\Theta(\log n)
\end{equation}
within the considered model. The converse is independent of the
particular encoding, decoding, recovery operation, or quantum code
employed at the intermediate stations.

The operational meaning of $m$ should be distinguished from the
asymptotic rates studied in general quantum-network capacity
frameworks. Here, $m$ is the number of physical channels allocated to every elementary
link. It need not only represent $m$ distinct transmission lines operating
simultaneously, it could represent the same physical channel being used sequentially $m$ times,
provided the transmitted systems are stored until all relevant physical systems are received. Such serialization does not improve the resource requirement, but transfers the bulk of the cost to long-lived quantum
memory.

The scope of the result is the homogeneous sequential architecture
defined above, with noiseless intermediate processing. It captures
the idealized logical layer of a broad family of direct one-way
error-correction-based repeaters, as discussed in the introduction,
but is not intended to describe all repeater architectures. In
particular, protocols based on separately generated elementary-link
entanglement, two-way purification, and entanglement swapping employ
a different resource model. This includes recent measurement-based
distillation architectures based on locally prepared resource states
and elementary-link Bell pairs \cite{Shirepeater}. Likewise,
genuinely bosonic encoding schemes \cite{Swainrepeater} are not directly
covered by our finite-dimensional theorems without an effective
finite-dimensional reduction.

\medskip

\emph{\textbf{Acknowledgments.}} This work was supported by the National Science Centre Poland (Grant No. 2022/46/E/ST2/00115 and 2024/55/B/ST2/01590) and within the QuantERA II Programme (Grant No. 2021/03/Y/ST2/00178, acronym ExTRaQT) that has received funding from the European Union's Horizon 2020 research and innovation programme under Grant Agreement No. 101017733.

\bibliography{literature}

\newpage

\section*{Supplemental material}

\subsection*{Concurrence}

For a two-qubit state \(\rho\) the concurrence is defined as~\cite{Hillentanglementofformation}
\begin{equation}
C(\rho) = \max \{ \xi_1 - \xi_2 - \xi_3 - \xi_4, 0 \},
\end{equation}
where \(\xi_i\) are the square roots of the eigenvalues of \(\rho \tilde{\rho}\), arranged in decreasing order, and
\begin{equation}
\tilde{\rho} = (\sigma_y \otimes \sigma_y)\rho^{*}(\sigma_y \otimes \sigma_y),
\end{equation}
where \(\sigma_y\) is the Pauli \(y\)-matrix and \(\rho^{*}\) denotes complex conjugation in the computational basis. For two-qubit states, the concurrence is directly related to the entanglement of formation via
\begin{equation}
E_{\mathrm{f}}(\rho)=h\left(\frac{1}{2}+\frac{1}{2}\sqrt{1-C(\rho)^{2}}\right),
\end{equation}
where \(h(x)=-x\log_{2}x-(1-x)\log_{2}(1-x)\) is the binary entropy function.

\subsection*{Proof of Theorem~\ref{thm::qubit_slocc}}
    In the proof, we employ the concurrence factorization law~\cite{Konradfactorization}, which states that for any completely positive trace-preserving map $\Lambda$ one has
\begin{equation}
C\bigl((\Lambda \otimes \openone)[\rho^{AB}]\bigr)
\leq
C\bigl((\Lambda \otimes \openone)[\phi^{+}]\bigr)\, C(\rho^{AB})
=
C(\Gamma)\, C(\rho^{AB}).
\label{eq::concurrencerelation}
\end{equation}

Our protocol consists of $n$ rounds. In each round, an SLOCC
instrument is applied first, followed by the channel $\Lambda$.
The SLOCC instrument used in the $i$-th round may depend on the complete
history of all outcomes obtained in the previous rounds.
Let $h_{i-1}=(x_1,\ldots,x_{i-1})$
denote a particular history before the $i$-th round. At the beginning
of this round, the parties share the normalized state
$\rho_{h_{i-1}}$.
Let $\left\{
    \mathcal{F}^{h_{i-1}}_{i,x_i}
    \right\}_{x_i}$
be the SLOCC instrument chosen after observing the history
$h_{i-1}$. The outcome $x_i$ occurs with conditional probability
\begin{equation}
    p(x_i\mid h_{i-1})
    =
    \operatorname{Tr}\left[
        \mathcal{F}^{h_{i-1}}_{i,x_i}
        \left(\rho_{h_{i-1}}\right)
    \right],
\end{equation}
and the corresponding normalized post-filter state is
\begin{equation}
    \sigma_{h_i}
    =
    \frac{
        \mathcal{F}^{h_{i-1}}_{i,x_i}
        \left(\rho_{h_{i-1}}\right)
    }{
        p(x_i\mid h_{i-1})
    },
    \qquad
    h_i=(h_{i-1},x_i).
\end{equation}

Since concurrence is an entanglement monotone, it cannot increase on
average under SLOCC. Therefore,
\begin{equation}
    \sum_{x_i}
    p\,(x_i|\,h_{i-1})\,
    C\left(\sigma_{h_i}\right)
    \leq
    C\left(\rho_{h_{i-1}}\right).
    \label{eq::filter}
\end{equation}

After the filter, the channel $\Lambda$ acts on one subsystem. The
state at the end of the $i$-th round is therefore
\begin{equation}
    \rho_{h_i}
    =
    \left(\Lambda\otimes\openone \right)
    \left(\sigma_{h_i}\right).
\end{equation}

We define
\begin{equation}
    q
    =
    C\left[
        \left(\Lambda\otimes\openone\right)
        \left(\phi^+\right)
    \right].
\end{equation}
The factorization law for concurrence implies, for every outcome
$x_i$,
\begin{equation}
    C\left(\rho_{h_i}\right)
    =
    C\left[
        \left(\Lambda\otimes\openone\right)
        \left(\sigma_{h_i}\right)
    \right]
    \leq
    q\,C\left(\sigma_{h_i}\right).
    \label{eq::channel}
\end{equation}

Multiplying Eq.~\eqref{eq::channel} by
$p(x_i\mid h_{i-1})$ and summing over all possible outcomes gives
\begin{equation}
    \sum_{x_i}
    p(x_i\mid h_{i-1})
    C\left(\rho_{h_i}\right)
    \leq
    q
    \sum_{x_i}
    p(x_i\mid h_{i-1})
    C\left(\sigma_{h_i}\right).
\end{equation}
Using Eq.~\eqref{eq::filter}, we obtain
\begin{equation}
    \sum_{x_i}
    p(x_i\mid h_{i-1})
    C\left(\rho_{h_i}\right)
    \leq
    q\,C\left(\rho_{h_{i-1}}\right).
    \label{eq::one_round_adaptive}
\end{equation}
Let $H_i=(X_1,\ldots,X_i)$ denote the random history after the
$i$-th round, and define the random variable
\begin{equation}
    C_i=C\left(\rho_{H_i}\right).
\end{equation}
For every fixed history $h_{i-1}$, the left-hand side of
Eq.~\eqref{eq::one_round_adaptive} is precisely the conditional
expectation
\begin{equation}
    \mathbb{E}\left[
        C_i\mid H_{i-1}=h_{i-1}
    \right]
    =
    \sum_{x_i}
    p(x_i\mid h_{i-1})
    C\left(\rho_{h_i}\right).
\end{equation}
Consequently,
\begin{equation}
    \mathbb{E}\left[
        C_i\mid H_{i-1}=h_{i-1}
    \right]
    \leq
    q\,C\left(\rho_{h_{i-1}}\right)
\end{equation}
for every possible history $h_{i-1}$. Equivalently,
\begin{equation}
    \mathbb{E}\left[
        C_i\mid H_{i-1}
    \right]
    \leq
    q C_{i-1}.
    \label{eq::conditional_concurrence}
\end{equation}

Taking the expectation of both sides with respect to the random
history $H_{i-1}$ and using the law of total expectations,
we find
\begin{align}
    \mathbb{E}[C_i]
    &=
    \mathbb{E}\left[
        \mathbb{E}\left[C_i\mid H_{i-1}\right]
    \right]
    \nonumber\\
    &\leq
    q\,\mathbb{E}[C_{i-1}].
\end{align}
Iterating this inequality over $n$ rounds yields
\begin{equation}
   \mathbb{E}[C_n]
    \leq
    q^n \mathbb{E}[C_0]=q^{n}C(\rho_0), 
\end{equation}
where in the last step we used a fact that we start with a fixed state so there is only one possibility, and therefore expectation value is equal to the concurrence value itself.
\begin{equation}
    \mathbb{E}[C_n]
    \leq
    q^n C(\rho_0).
    \label{eq::average_concurrence_bound}
\end{equation}

Finally, since $C_n$ is a non-negative random variable, Markov's
inequality gives, for every $c>0$,
\begin{align}
    \Pr\left[C_n\geq c\right]
    &\leq
    \frac{\mathbb{E}[C_n]}{c}
    \nonumber\\
    &\leq
    \frac{q^n C(\rho_0)}{c}
    \nonumber\\
    &\leq
    \frac{q^n}{c}.
\end{align}

\subsection*{Proof of Corollary~\ref{thm::qubit}}
In this section we present a proof of Corollary \ref{thm::qubit} using factorization law from \cite{Konradfactorization}.
We prove the claim by induction. Specifically, we establish that
\begin{equation}
C\bigl((\openone \otimes \Lambda_n)[\rho^{AB}]\bigr) \leq C(\Gamma)^n \, C(\rho^{AB}).
\label{eq::inductionqubits}
\end{equation}
To do so, we employ the factorization law for concurrence, which extends to mixed two-qubit states in the form~\cite{Konradfactorization}
\begin{equation}
C\bigl((\Lambda \otimes \openone)[\rho^{AB}]\bigr)
\leq
C\bigl((\Lambda \otimes \openone)[\phi^{+}]\bigr)\, C(\rho^{AB})
=
C(\Gamma)\, C(\rho^{AB}).
\end{equation}

For the base case \(n = 1\), we obtain
\begin{equation}
C\bigl((\Lambda \otimes \openone)\circ F_1[\rho^{AB}]\bigr)
\leq
C(\Gamma)\, C\bigl(F_1[\rho^{AB}]\bigr)
\leq
C(\Gamma)\, C(\rho^{AB}),
\end{equation}
where \(F_1\) is an LOCC filter. In the second inequality, we used the fact that LOCC operations cannot increase concurrence. We have thus established that Eq.~\eqref{eq::inductionqubits} holds for \(n = 1\). 

For the induction step, note that
\begin{equation}
\Lambda_{n+1}[\rho^{AB}]
=
(\Lambda \otimes \openone)\circ F_{n+1}\circ \Lambda_n[\rho^{AB}]
=
(\Lambda \otimes \openone)\circ F_{n+1}[\sigma^{AB}],
\end{equation}
where we defined the state $\sigma^{AB} = \Lambda_n[\rho^{AB}]$. Now assume that Eq.~\eqref{eq::inductionqubits} holds for some \(n\). For \(n+1\), we have
\begin{align}
C\bigl(\Lambda_{n+1}[\rho^{AB}]\bigr)
&=
C\bigl((\Lambda \otimes \openone)\circ F_{n+1}[\sigma^{AB}]\bigr) \\
&\leq
C(\Gamma)\, C\bigl(F_{n+1}[\sigma^{AB}]\bigr) \nonumber \\
&\leq
C(\Gamma)\, C(\sigma^{AB}). \nonumber
\end{align}
By the induction hypothesis,
\begin{equation}
C(\sigma^{AB}) \leq C(\Gamma)^n \, C(\rho^{AB}).
\end{equation}
Combining the above inequalities, we obtain
\begin{equation}
C\bigl(\Lambda_{n+1}[\rho^{AB}]\bigr)
\leq
C(\Gamma)^{n+1} \, C(\rho^{AB}),
\end{equation}
which completes the induction.

\subsection*{Proof of Theorem~\ref{thm:General}}
\begin{proof}
It is clear to see that if $\Lambda$ contains a correctable subspace then Eq.~(\ref{eq:Theorem2Inequality}) holds. We will now deal with the reverse implication.

If Eq.~(\ref{eq:Theorem2Inequality}) holds, then there must exist an entangled state \( \rho^{AB} \) such that 
\begin{equation} \label{eq:ConstantEf}
E_{\mathrm{f}}(\Lambda\otimes\openone[\rho^{AB}])=E_{\mathrm{f}}(\rho^{AB}).
\end{equation}
See Lemma~\ref{lem:StableState} in the supplemental material for proof. We now show that this condition implies the existence of a pure entangled state
\( \ket{\psi^{AB}} \) satisfying
\begin{equation}
E_{\mathrm{f}}(\Lambda\otimes\openone[\psi^{AB}])=E_{\mathrm{f}}(\psi^{AB}).\label{eq:Pure}
\end{equation}

For this, let $\{p_{i},\psi_{i}^{AB}\}$ be an optimal pure state
decomposition of $\rho^{AB}$, i.e.,
\begin{equation}
E_{\mathrm{f}}(\rho^{AB})=\sum_{i}p_{i}E_{\mathrm{f}}(\psi_{i}^{AB}).
\end{equation}
Since $E_{\mathrm{f}}$ does not increase under local operations,
for each of the states $\ket{\psi_{i}^{AB}}$ we have 
\begin{equation}
E_{\mathrm{f}}(\Lambda\otimes\openone[\psi_{i}^{AB}])\leq E_{\mathrm{f}}(\psi_{i}^{AB}).
\end{equation}
By contradiction, assume that for all entangled states $\ket{\psi_{j}^{AB}}$
this inequality is strict, i.e., 
\begin{equation}
E_{\mathrm{f}}(\Lambda\otimes\openone[\psi_{j}^{AB}])<E_{\mathrm{f}}(\psi_{j}^{AB}).
\end{equation}
Using these assumptions and convexity~\cite{Vidalentanglementmonotones} of $E_{\mathrm{f}}$ we obtain
\begin{align}
E_{\mathrm{f}}(\Lambda\otimes\openone[\rho^{AB}]) & =E_{\mathrm{f}}\left(\sum_{i}p_{i}\Lambda\otimes\openone[\psi_{i}^{AB}]\right)\\
 & \leq\sum_{i}p_{i}E_{\mathrm{f}}(\Lambda\otimes\openone[\psi_{i}^{AB}])\nonumber \\
 & <\sum_{i}p_{i}E_{\mathrm{f}}(\psi_{i}^{AB})=E_{\mathrm{f}}(\rho^{AB}).\nonumber 
\end{align}
This contradicts Eq.~(\ref{eq:ConstantEf}), which implies that Eq.~(\ref{eq:Pure})
is true for some entangled state $\ket{\psi}^{AB}$.

Let now $\ket{\psi}^{AB}$ be an entangled state such that Eq.~(\ref{eq:Pure}) holds.
In the following, let $\{K_{i}\}$ be a set of Kraus operators corresponding
to the quantum channel $\Lambda$. We further define the probabilities
$q_{i}$ and pure quantum states $\phi_{i}^{AB}$ as follows: 
\begin{align}
q_{i} & =\mathrm{Tr}\left[K_{i}\otimes\openone\psi^{AB}K_{i}^{\dagger}\otimes\openone\right],\\
\phi_{i}^{AB} & =\frac{1}{q_{i}}K_{i}\otimes\openone\psi^{AB}K_{i}^{\dagger}\otimes\openone.
\end{align}
Without loss of generality we assume that $q_{i}>0$ for all $i$. 

Consider now the average entanglement of the ensemble $\{q_{i},\phi_{i}^{AB}\}$:
\begin{equation}
\sum_{i}q_{i}E_{\mathrm{f}}(\phi_{i}^{AB})=\sum_{i}q_{i}S(\phi_{i}^{B}).
\end{equation}
By the concavity of the von Neumann entropy it holds that 
\begin{equation} \label{eq:Concavity}
\sum_{i}q_{i}S(\phi_{i}^{B})\leq S\left(\sum_{i}q_{i}\phi_{i}^{B}\right)=S(\psi^{B}).
\end{equation}
Since the von Neumann entropy is strictly concave~\cite{Nielsen_Chuang_2010}, Eq.~(\ref{eq:Concavity}) holds with equality if and only if 
\begin{equation}
\phi_{i}^{B}=\psi^{B}\,\,\,\mathrm{for\,all}\,i.
\end{equation}
By convexity of the entanglement of formation we further have 
\begin{equation}
E_{\mathrm{f}}(\Lambda\otimes\openone[\psi^{AB}])\leq\sum_{i}q_{i}E_{\mathrm{f}}(\phi_{i}^{AB})=\sum_{i}q_{i}S(\phi_{i}^{B}).
\end{equation}
Summarizing these arguments, the condition~(\ref{eq:Pure}) implies
that $\phi_{i}^{B}=\psi^{B}$ holds for all $i$. A similar argument has also been used in~\cite{PhysRevA.99.022338}.

In the next step we introduce a reference system $R$.
By the Stinespring dilation theorem~\cite{Nielsen_Chuang_2010}, for any quantum channel $\Lambda$ there exists a unitary
$U_{RA}$ on $RA$ such that 
\begin{equation}
\Lambda\otimes\openone[\psi^{AB}]=\mathrm{Tr}_{R}\left[U_{RA}\left(\ket{0}\!\bra{0}^{R}\otimes\psi^{AB}\right)U_{RA}^{\dagger}\right]
\end{equation}
with a pure state $\ket{0}^{R}$. We further define the pure state
\begin{equation}
\ket{\mu}^{RAB}=U_{RA}\left(\ket{0}^{R}\otimes\ket{\psi}^{AB}\right).
\end{equation}

Note that any rank-1 POVM $\{M_{j}^{R}\}$ on the reference system
$R$ gives rise to a set of Kraus operators $\{L_{j}\}$ for the channel
$\Lambda$, such that
\begin{equation}
\mathrm{Tr}_{R}\left[M_{j}^{R}\mu^{RAB}\right]=\left(L_{j}\otimes\openone\right)\psi^{AB}\left(L_{j}^{\dagger}\otimes\openone\right).
\end{equation}
From the arguments given above we see that Eq.~(\ref{eq:Pure}) implies
that the following equality holds for all rank-1 POVMs $\{M_{j}^{R}\}$
and all $j$:
\begin{equation}
\mathrm{Tr}_{RA}\left[M_{j}^{R}\mu^{RAB}\right]=\mathrm{Tr}_{RAB}\left[M_{j}^{R}\mu^{RAB}\right]\times\psi^{B}.\label{eq:Proof-1}
\end{equation}
We will now take a closer look at the state $\mu^{RB}$. Due to Eq.~(\ref{eq:Proof-1}),
it holds that 
\begin{equation}
\mathrm{Tr}_{R}\left[M_{j}^{R}\mu^{RB}\right]=\mathrm{Tr}_{RB}\left[M_{j}^{R}\mu^{RB}\right]\times\psi^{B}
\end{equation}
for all POVMs $\{M_{j}^{R}\}$ and all $j$. As we show in Lemma~\ref{prop:Product} in the supplemental material,
this is only possible if $\mu^{RB}$ is a product state, i.e., 
\begin{equation}
\mu^{RB}=\mu^{R}\otimes\psi^{B}.\label{eq:Product}
\end{equation}

Recalling that the state $\mu^{RAB}$ is pure and using Eq.~(\ref{eq:Product}),
we see that it is always possible to attach an ancillary system $A'$
on Alice's side in a pure state $\ket{0}^{A'}$ and perform a unitary
$U_{AA'}$ on Alice's part of the total state $\ket{\mu}^{RAB}\otimes\ket{0}^{A'}$
such that 
\begin{equation}
U_{AA'}\left(\ket{\mu}^{RAB}\otimes\ket{0}^{A'}\right)=\ket{\nu}^{RA'}\otimes\ket{\psi}^{AB},
\end{equation}
where $\ket{\nu}^{RA'}$ is a purification of $\mu^{R}$. This implies
that the state $\mu^{AB}\otimes\ket{0}\!\bra{0}^{A'}$ is transformed
as follows:
\begin{equation}
U_{AA'}\left(\mu^{AB}\otimes\ket{0}\!\bra{0}^{A'}\right)U_{AA'}^{\dagger}=\ket{\psi}\!\bra{\psi}^{AB}\otimes\nu^{A'}.
\end{equation}

Recalling that $\mu^{AB}=\Lambda\otimes\openone[\psi^{AB}]$, these arguments imply that there exists a quantum operation $\Lambda'$ on
Alice's subsystem such that 
\begin{equation}
(\Lambda'\circ\Lambda)\otimes\openone[\psi^{AB}]=\psi^{AB}.
\end{equation}
Since $\ket{\psi}^{AB}$ is entangled, it has a Schmidt decomposition
of the form $\ket{\psi}^{AB}=\sum_{i}\sqrt{\lambda_{i}}\ket{ii}$
with at least two nonzero Schmidt coefficients $\lambda_i$. Consider now a pure state $\ket{\eta}$ of the form 
\begin{equation}
\ket{\eta}=\sum_{i}c_{i}\ket{i}^{A},
\end{equation}
where $c_i$ are complex coefficients and $\ket{i}^A$ are local basis states on Alice's subsystem for which the state $\ket{\psi}^{AB}$ has a nonzero Schmidt coefficient. Using Lemma~\ref{prop:CorrectableSubspace} in the supplemental material, it follows that
\begin{equation}
\Lambda'\circ\Lambda[\eta]=\eta
\end{equation}
and the proof of the theorem is complete.
\end{proof}

\begin{lem} \label{lem:StableState}
    Let $\Lambda$ be a quantum channel and $\rho^{AB}$ an initial state. If
    \begin{equation}
        \lim_{n\rightarrow\infty}E_\mathrm f(\Lambda_{n}[\rho^{AB}])>0,
    \end{equation}
    then there exists an entangled state $\mu^{AB}$ such that 
    \begin{equation}
        E_{\mathrm{f}}(\Lambda\otimes\openone[\mu^{AB}])=E_{\mathrm{f}}(\mu^{AB}).
    \end{equation}
\end{lem}

\begin{proof}
Let
\begin{equation}
c = \lim_{n \to \infty} E_{\mathrm f}\bigl(\Lambda_n[\rho^{AB}]\bigr),
\end{equation}
so that \(c > 0\) by assumption. Then, for every \(\varepsilon > 0\), there exists \(N\) such that
\begin{equation}
\left| E_{\mathrm f}\bigl(\Lambda_k[\rho^{AB}]\bigr) - c \right| < \varepsilon \label{eq:EfBound-1}
\end{equation}
for all \(k > N\). In particular, for all \(k > N\) we also have
\begin{equation}
\left| E_{\mathrm f}\bigl(\Lambda_{k+1}[\rho^{AB}]\bigr) - c \right| < \varepsilon. \label{eq:EfBound-2}
\end{equation}

Now observe that
\begin{equation}
\Lambda_{k+1}[\rho^{AB}]
=
(\Lambda \otimes \openone)\circ F_{k+1}\circ \Lambda_k[\rho^{AB}]
=
(\Lambda \otimes \openone)[\sigma^{AB}], \label{eq:LambdaK1}
\end{equation}
where we define
\begin{equation}
\sigma^{AB} = F_{k+1}\circ \Lambda_k[\rho^{AB}].
\end{equation}
Since the entanglement of formation is non-increasing under LOCC, it follows that
\begin{equation}
E_{\mathrm f}\bigl(\Lambda_{k+1}[\rho^{AB}]\bigr)
\leq
E_{\mathrm f}(\sigma^{AB})
\leq
E_{\mathrm f}\bigl(\Lambda_k[\rho^{AB}]\bigr).
\end{equation}
Combining this with Eqs.~\eqref{eq:EfBound-1} and \eqref{eq:EfBound-2}, we obtain
\begin{equation}
c - \varepsilon < E_{\mathrm f}(\sigma^{AB}) < c + \varepsilon,
\end{equation}
or equivalently,
\begin{equation}
\left| E_{\mathrm f}(\sigma^{AB}) - c \right| < \varepsilon.
\end{equation}
    
Using the triangle inequality together with Eq.~\eqref{eq:LambdaK1}, we therefore obtain
\begin{align}
\left|E_{\mathrm{f}}(\Lambda\otimes\openone[\sigma^{AB}])-E_{\mathrm{f}}(\sigma^{AB})\right| & \leq\left|E_{\mathrm{f}}(\Lambda\otimes\openone[\sigma^{AB}])-c\right|\\
& +\left|E_{\mathrm{f}}(\sigma^{AB})-c\right|<2\varepsilon.\nonumber 
\end{align}
Since \(\varepsilon > 0\) can be chosen arbitrarily small, the proof is complete.
\end{proof}

\begin{lem}
\label{prop:Product}A quantum state $\rho^{AB}$ fulfills 
\begin{equation}
\mathrm{Tr}_{A}\left[M_{i}^{A}\rho^{AB}\right]=\mathrm{Tr}_{AB}\left[M_{i}^{A}\rho^{AB}\right]\times\rho^{B}
\label{eq::propositionsteering}
\end{equation}
for all POVMs $\{M_{i}^{A}\}$ on the subsystem $A$ and all $i$
if and only if $\rho^{AB}$ is of the form 
\begin{equation}
\rho^{AB}=\rho^{A}\otimes\rho^{B}.
\end{equation}
\end{lem}

\begin{proof} 
The implication $\Leftarrow$ is immediate. We now prove the implication $\Rightarrow$. Let us write a state $\rho^{AB}$ in its Operator-Schmidt decomposition~\cite{KhatriPQCT} such that $\rho^{AB}=\sum_k \sqrt{\lambda_k} E_k\otimes F_k,$ where $\{E_k\}$ and $\{F_k\}$ form orthonormal sets of Hermitian operators acting on systems $A$ and $B$, respectively. Without loss of generality, we restrict ourselves to $\lambda_k\ne 0.$ We also expand the operator $M_i^{A}$ in the basis $\{E_\alpha\}$ as $M_i^{A}=\sum_{\alpha}m_{\alpha}E_{\alpha},$ where $m_\alpha$ are real coefficients. We have:
\begin{equation}
\begin{split}
    &\mathrm{Tr}_{A}\left[M_{i}^{A}\rho^{AB}\right]=\sum_{k,\alpha}\sqrt{\lambda_k}m_{\alpha}\mathrm{Tr}\left(E_k E_{\alpha}\right)F_k\\
    &=\sum_{k}\sqrt{\lambda_k}m_{k}F_k=\mathrm{Tr}\left[M_{i}^{A}\rho^{A}\right]\sum_k b_kF_k,
\end{split}   
\label{eq::productexpansion}
\end{equation}
where in the last equality we used Eq.~\eqref{eq::propositionsteering} and $\mathrm{Tr}\left[M_i^{A}\rho^{AB}\right]=\mathrm{Tr}\left[M_{i}^{A}\rho^{A}\right].$ We also decomposed $\rho^{B}$ in a basis $F_k$ obtaining $\rho^{B}=\sum_k b_kF_k$ with real coefficients $b_k$. Coefficient $m_k$ satisfies $m_k=\mathrm{Tr}\left(M_i^{A}E_k\right).$ Using this, and comparing coefficients in Eq.~\eqref{eq::productexpansion} we obtain:

\begin{equation}
    \mathrm{Tr}\left(M_i^{A}E_k\right)=\frac{b_k}{\sqrt{\lambda_k}}\mathrm{Tr}\left(M_i^{A}\rho^{A}\right).
\end{equation}
It must hold for all $M_i^{A},$ thus $E_k=\frac{b_k}{\sqrt{\lambda_k}}\rho^{A}.$ Plugging it into $\rho^{AB}$ we obtain:
\begin{equation}
    \rho^{AB}=\sum_k \sqrt{\lambda_k} E_k\otimes F_k=\sum_k \rho^{A}\otimes b_k F_k=\rho^{A}\otimes \rho^{B}.
\end{equation}
\end{proof}

\begin{lem} \label{prop:CorrectableSubspace}
Let $\ket{\psi}^{AB}$ be a pure entangled state satisfying
\begin{equation}\label{eqn:recovery_condition}
    (N \otimes \openone)\bigl[\psi^{AB}\bigr] = \psi^{AB},
\end{equation}
where $N$ is a quantum channel and
\begin{equation}
    \ket{\psi}^{AB} = \sum_i \sqrt{\lambda_i} \ket{i}^{A} \ket{i}^{B},
\end{equation}
such that $\ket{\psi}^{AB}$ has at least two non-zero Schmidt coefficients. Here, $\{\ket{i}^{A}\}$ and $\{\ket{i}^{B}\}$ denote the Schmidt bases for subsystems $A$ and $B$, respectively.

Then, for any pure state $\ket{\eta}$ belonging to the subspace spanned by $\{\ket{i}^{A}\}$, i.e.
\begin{equation}
\ket{\eta} = \sum_i c_i \ket{i}^{A},
\end{equation}
with $c_i \in \mathbb{C}$, it holds that
\begin{equation}
    N[\eta] = \eta.
\end{equation}
\end{lem}

\begin{proof}

From Eq.~\eqref{eqn:recovery_condition} and the assumed Schmidt decomposition of $\ket{\psi}^{AB}$, we obtain
\begin{align}
\sum_{i,j}\sqrt{\lambda_i \lambda_j}\,
N\!\left[\ket{i}\!\bra{j}^{A}\right]\otimes \ket{i}\!\bra{j}^{B}
=
\sum_{i,j}\sqrt{\lambda_i \lambda_j}\,
\ket{i}\!\bra{j}^{A}\otimes \ket{i}\!\bra{j}^{B}.
\end{align}

Left multiplying by $\openone\otimes \ket{k}\!\bra{l}$ and tracing over subsystem $B$, we obtain
\begin{equation}
N\!\left[\ket{l}\!\bra{k}^{A}\right] = \ket{l}\!\bra{k}^{A}.
\end{equation}

Since $N$ is a quantum channel, it is linear on operators. This completes the proof.   
\end{proof}

\subsection*{Proof of Theorem~\ref{thm:exponential}}
In the following we consider a pure bipartite state $\ket{\psi}^{AB}=\sum_{i=0}^{k-1}\sqrt{\lambda_{i}}\ket{ii}$
with Schmidt rank at most $k$~\cite{Schmidtdecomposition1a2bstate,Nielsen_Chuang_2010}, such that the Schmidt coefficients
$\lambda_{i}$ are sorted in non-increasing order. We define the $G_{k}$-concurrence as, (similar quantity was also defined in ~\cite{GourGC} )
\begin{equation}
G_{k}(\psi^{AB})=k \left(\prod_{i=0}^{k-1}\lambda_{i}\right)^{1/k}.
\end{equation}
This definition can be extended in a straightforward way also to non-normalized
vectors in $\mathcal{H}_{AB}$. We further define the state 
\begin{equation}
\ket{\phi}_{k}^{+}=\frac{1}{\sqrt{k}}\sum_{i=0}^{k-1}\ket{ii}
\end{equation} 
to be maximally entangled state of rank $k.$

Let us consider a quantum channel $\Lambda.$ We define:
\begin{equation}
    \kappa=\min_{K_i}\sum_iG_k[(K_i\otimes \openone)\,\phi_k^{+}(K_i^{\dagger}\otimes \openone)],
\end{equation}
where the minimization is taken over all Kraus decompositions of $\Lambda.$ The $\kappa$ parameter quantifies how good is channel in entanglement preservation.

\begin{lem}
\label{lem:Gk}For any pure state $\ket{\psi}^{AB}$ with Schmidt
rank at most $k$ and any matrix $M$ on $\mathcal{H}_{A},$ the following
equality holds:
\begin{multline}
G_{k}[(M\otimes\openone)\,\psi^{AB}(M^{\dagger}\otimes\openone)]\\=G_{k}[(M\otimes\openone)\,\phi_{k}^{+}(M^{\dagger}\otimes\openone)]\times G_{k}(\psi^{AB}).
\label{eq::klemma}
\end{multline}
\end{lem}

\begin{proof}

If $k$ is greater than Schmidt rank of $\ket{\psi^{AB}},$ both sides of the Eq.~\eqref{eq::klemma} are 0. Assume that Schmidt rank of $\ket{\psi^{AB}}$ is $k$. We can write our state as~\cite{KhatriPQCT}:
\begin{equation}
    \ket{\psi^{AB}}=\sum_{i,j} X_{i,j}\ket{i}_A\ket{j}_B=\sum_{i=1}^{k}\sqrt{\lambda_i}\ket{e_i}_A\ket{f_i}_B,
    \label{eq::klemmade}
\end{equation}
where $\ket{i}_A, \ket{e_i}_A$ are orthonormal bases in space $\mathcal{H}_{A},$ $\ket{j}_B, \ket{f_i}_B$ are orthonormal bases in space $\mathcal{H}_{B}$ and $\sqrt{\lambda_i}$ are singular values of $X$. We further have $M\otimes \openone \ket{\psi^{AB}}=\sum_{i,j} (M X)_{i,j}\ket{i}_A\ket{j}_B.$ Thus, the Schmidt coefficients of $M\otimes \openone \ket{\psi^{AB}}$ are singular values of $M X$.  

Let
\begin{equation}
    \mathcal S_A
    =
    \operatorname{span}\{\ket{e_1},\ldots,\ket{e_k}\},
    \qquad
    \mathcal S_B
    =
    \operatorname{span}\{\ket{f_1},\ldots,\ket{f_k}\}
\end{equation}
be the Schmidt supports on subsystems $A$ and $B$,
respectively. Let
\begin{equation}
    X_k
    =
    \operatorname{diag}
    \left(
        \sqrt{\lambda_1},\ldots,\sqrt{\lambda_k}
    \right)
\end{equation}
be the coefficient matrix represented in the corresponding
Schmidt bases, and let
\begin{equation}
    M_k:=M|_{\mathcal S_A}:
    \mathcal S_A\rightarrow\mathcal H_A
\end{equation}
be the restriction of $M$ to the Schmidt support. The
nonzero singular values of $MX$ coincide with the singular
values of $M_kX_k$. Singular values of $M_kX_k$ satisfy~\cite{Horn_Johnson_1991}:
\begin{equation}
\begin{split}
    &\prod_{i=1}^{k}\sigma_i(M_kX_k)
    =
    \sqrt{
        \det\left(
            X_k^{\dagger}M_k^{\dagger}M_kX_k
        \right)
    }\\
    &=
    \sqrt{\det(M_k^{\dagger}M_k)}
    \sqrt{\det(X_k^{\dagger}X_k)}=
    \prod_{i=1}^{k}s_i
    \prod_{i=1}^{k}\sqrt{\lambda_i},
\end{split}
\label{eq::klemmasingularequality}
\end{equation}
where $s_i$ are singular values for a matrix $M_k.$ Using above relation we can obtain formula for $G_k$-concurrence.
\begin{multline}
    G_{k}[(M\otimes\openone)\,\psi^{AB}(M^{\dagger}\otimes\openone)]=k\left(\prod_{i=1}^{k}s_i^{2}\prod_{i=1}^{k}\lambda_i\right)^{1/k}\\=\left(\prod_{i=1}^{k}s_i^{2}\right)^{1/k}G_k(\psi^{AB}).
    \label{eq::lemmakfactoring}
\end{multline}
Substituting $\psi^{AB}= \phi^{+}_{k}$ in Eq.~\eqref{eq::lemmakfactoring}, we see that 
\begin{equation}
    G_{k}[(M\otimes\openone)\,\phi^{+}_{k}(M^{\dagger}\otimes\openone)]=\left(\prod_{i=1}^{k}s_i^{2}\right)^{1/k}.
\end{equation}
Plugging it into Eq.~\eqref{eq::lemmakfactoring} finishes the proof.
\end{proof}
\begin{prop} \label{prop:SchmidtRankK}
For a pure state $\ket{\psi}^{AB}$ with Schmidt rank at most $k$
and a local quantum operation $\Lambda$ which has no correctable
subspace the state $\Lambda\otimes\openone[\psi^{AB}]$ can be written
as 
\begin{equation}
\Lambda\otimes\openone[\psi^{AB}]=\sum_{i}p_{i}\ket{\phi_{i}}\!\bra{\phi_{i}}^{AB}
\end{equation}
where each state $\ket{\phi_{i}}^{AB}$ has Schmidt rank at most
$k$ and 
\begin{equation}
\sum_{i}p_{i}G_{k}(\phi_{i}^{AB})=\kappa_s G_{k}(\psi^{AB})
\end{equation}
with $\kappa_s<1$.
\end{prop}

\begin{proof}
Let $\{K_{i}\}$ be a set of Kraus operators for $\Lambda$. Using
Lemma~\ref{lem:Gk}, we have
\begin{align}
 & \sum_{i}G_{k}[(K_{i}\otimes\openone)\psi^{AB}(K_{i}^{\dagger}\otimes\openone)]\\
 & =\sum_{i}G_{k}[(K_{i}\otimes\openone)\phi_{k}^{+}(K_{i}^{\dagger}\otimes\openone)]\times G_{k}(\psi^{AB}).\nonumber 
\end{align}
As we prove in Lemma~\ref{lem:GkBound}, for any quantum channel $\Lambda$ which has no correctable subspace, the set of Kraus operators $\{K_{i}\}$
can always be chosen such that
\begin{equation}
\sum_{i}G_{k}[(K_{i}\otimes\openone)\phi_{k}^{+}(K_{i}^{\dagger}\otimes\openone)]<1.
\end{equation}
Defining
\begin{equation}
\kappa_s=\sum_{i}G_{k}[(K_{i}\otimes\openone)\phi_{k}^{+}(K_{i}^{\dagger}\otimes\openone)]
\label{eq::kappa}
\end{equation}
we obtain 
\begin{equation}
\sum_{i}G_{k}[(K_{i}\otimes\openone)\psi^{AB}(K_{i}^{\dagger}\otimes\openone)]=\kappa_s G_{k}(\psi^{AB}).
\end{equation}
The proof of the proposition is then completed by defining probabilities
$p_{i}$ and pure states $\ket{\phi_{i}}^{AB}$ as follows:
\begin{align}
p_{i} & =\mathrm{Tr}\left[K_{i}\otimes\openone\psi^{AB}K_{i}^{\dagger}\otimes\openone\right],\\
\ket{\phi_{i}}^{AB} & =\frac{1}{\sqrt{p_{i}}}K_{i}\otimes\openone\ket{\psi}^{AB},
\end{align}
and noting that
\begin{equation}
\sum_{i}G_{k}[(K_{i}\otimes\openone)\psi^{AB}(K_{i}^{\dagger}\otimes\openone)]=\sum_{i}p_{i}G_{k}(\phi_{i}^{AB}).
\end{equation}
\end{proof}

\begin{lem} \label{lem:GkBound}
    Any quantum channel $\Lambda$ which has no correctable subspace has a set of Kraus operators $\{K_{i}\}$ such that that
    \begin{equation}
    \kappa_s=\sum_{i}G_{k}[(K_{i}\otimes\openone)\phi_{k}^{+}(K_{i}^{\dagger}\otimes\openone)]<1
    \end{equation}
    for any $k \geq 2$.
\end{lem}
\begin{proof}
We begin by showing that if $\Lambda$ does not have a correctable subspace, then $E_f(\Lambda \otimes \openone \ket{\phi_k^{+}})<E_f(\ket{\phi_k^{+}})$. We proceed by contradiction. Assume that there is no correctable subspace and that $E_f(\Lambda \otimes \openone \ket{\phi_k^{+}})=E_f(\ket{\phi_k^{+}})$ for some $k$. Repeating the steps of the proof of Theorem~\ref{thm:General} with $\phi_k^{+}$ in place of $\psi^{AB},$ we obtain that a correctable subspace must exist. This contradicts the assumption that no correctable subspace exists. Therefore, $E_f(\Lambda \otimes \openone \ket{\phi_k^{+}})<E_f(\ket{\phi_k^{+}}).$

Let $\rho=\Lambda \otimes \openone(\ket{\phi_k^{+}}).$ Let us denote by $\{q_i, \ket{\psi_i}\}$ the ensemble decomposition of $\rho$ that minimizes $\sum_iq_iE_f(\ket{\psi_i}).$ There exist Kraus operators $\{L_i\}$ such that $L_i\otimes \openone \ket{\phi_k^{+}}\bra{\phi_k^{+}}L_i^{\dagger}\otimes \openone=q_i\ket{\psi_i}\bra{\psi_i}.$ Additionally, we have $\sum_iq_iE_f(\ket{\psi_i})<E_f(\ket{\phi_k^{+}}).$ Consequently, for at least one $i$ there is inequality $E_f(\ket{\psi_i})<E_f(\ket{\phi_k^{+}}).$ Both the entanglement of formation and the $G_k$-concurrence reach their global maximum for pure states if and only if all their Schmidt coefficients are equal. Therefore, any pure state $\ket{\psi_i}$ that is not maximally entangled (i.e. satisfying $E_f(\ket{\psi_i})<E_f(\ket{\phi_k^{+}})$) must necessarily satisfy $G_k(\ket{\psi_i})<G_k(\ket{\phi_k^{+}})$. Consequently, since for at least one $i$ there is relation $G_k(\ket{\psi_i})<G_k(\ket{\phi_k^{+}}),$ the following inequality must hold
    \begin{equation}
        \sum_i q_iG_k(\ket{\psi_i})<G_k(\ket{\phi_k^{+}})=1
    \end{equation}
This concludes the proof.
\end{proof}

Consider now a mixed state $\rho^{AB}$ which can be written as $\rho^{AB}=\sum_{j}p_{j}\psi_{j}^{AB}$
with each of the pure states $\ket{\psi_{j}}^{AB}$ having Schmidt
rank at most $k$. Using Proposition~\ref{prop:SchmidtRankK}, each of the states $\Lambda\otimes\openone[\psi_{j}^{AB}]$
can be expressed as
\begin{equation}
\Lambda\otimes\openone[\psi_{j}^{AB}]=\sum_{i}q_{i}\psi_{ij}^{AB},
\end{equation}
where the pure states $\ket{\psi_{ij}}^{AB}$ have Schmidt rank at
most $k$ and 
\begin{equation}
\sum_{i}q_{i}G_{k}(\psi_{ij}^{AB})=\kappa_s G_{k}(\psi_{j}^{AB}).
\end{equation}
This means that state $\Lambda\otimes\openone[\rho^{AB}]$ can
be written as a convex combination of states $\ket{\psi_{ij}}^{AB}$
with probabilities $q_{i}p_{j}$ such that 
\begin{equation}
\sum_{ij}q_{i}p_{j}G_{k}(\psi_{ij}^{AB})=\kappa_s\left(\sum_{j}p_{j}G_{k}(\psi_{j}^{AB})\right).
\end{equation}
With these results, we can now prove the following proposition.

\begin{prop}{\label{pr::expon}}
Let $\Lambda_n$ be defined as in Eq. \eqref{eq:LambdaN} and let $\rho^{AB}$ be a mixed state with Schmidt rank at most $k.$ Then, $\Lambda_{n}[\rho^{AB}]=\sum_i s_i\phi_i^{AB},$ where each $\phi_i^{AB}$ has Schmidt rank at most $k$ and $\sum_i s_iG_k(\phi_i^{AB})\leq \kappa_s^{n},$ where $\kappa_s$ is defined in Eq. \eqref{eq::kappa}.
\end{prop}

\begin{proof}
We will prove this proposition by induction.
First, we show that the above conjecture holds for $n=1$. Let $\{K_{i}\}$ be a set of Kraus operators for $\Lambda$ and $\{M^{i}_{j}\otimes N^{i}_{k}\}$ be a set of Kraus operators for $F_{i}$. Let $\rho^{AB}=\sum_{j}q_j\ket{\phi_j^{AB}}\bra{\phi_j^{AB}},$ where each $\ket{\phi_j^{AB}}$ has Schmidt rank at most $k.$ For $n=1$ we have:
\begin{equation}
    \Lambda_{1}[\rho^{AB}] = (\Lambda\otimes\mathbb{I})[F_1\circ \rho^{AB}]=\sum_{ijlg}r_{lgi}q_j\phi_{ijgl}^{AB},
\end{equation}
such that
\begin{align}
    &\ket{\phi_{lgij}}^{AB} = \frac{1}{\sqrt{r_{lgi}}} (K_{i}\otimes\mathbb{I})(M_{l}\otimes N_{g})\,\ket{\phi_{j}}^{AB},\\
    \label{eq::ketphi}
    &r_{lgi}= \operatorname{Tr}[(K_{i}\otimes\mathbb{I})(M_{l}\otimes N_{g})\,\phi_{j}^{AB}(M_{l}\otimes N_{g})^{\dagger}(K_{i}\otimes\mathbb{I})^{\dagger}].
\end{align}
Note that each $\ket{\phi_{lgij}}^{AB}$ cannot have higher Schmidt rank than $\ket{\phi_j^{AB}}.$ Defining new indices $h$ as $h=ijgl$ we obtain:
\begin{equation}
    \Lambda_1(\rho^{AB})=\sum_h p_h \phi_h^{AB},
\end{equation}
where each $\phi_h^{AB}$ has Schmidt rank at most $k.$ 

We have shown the first part of the Proposition for $n=1.$ We proceed to the second part for $n=1.$
\begin{align}
    &\sum_{ijgl}\,q_i\,r_{lgi}\,G_{k}(\phi_{lgij}^{AB})\notag\\
    &= \sum_{ijgl}q_i\,G_{k}((K_{i}\otimes\mathbb{I})(M_{l}\otimes N_{g})\,\phi_{ij}^{AB}(M_{l}\otimes N_{g})^{\dagger}(K_{i}\otimes\mathbb{I})^{\dagger}).
    \label{eq::gkinkappaprooof1}
\end{align}
Now, we recall that, see Lemma~\ref{lem:Gk} 
\begin{align}
    &G_{k}((K_{i}\otimes\mathbb{I})(M_{l}\otimes N_{g})\,\phi_{ij}^{AB}(M_{l}\otimes N_{g})^{\dagger}(K_{i}\otimes\mathbb{I})^{\dagger})\notag \\ 
    &=G_{k}((K_{i}\otimes\mathbb{I})\phi^{+}_{k}(K_{i}\otimes\mathbb{I})^{\dagger})\,G_{k}((M_{l}\otimes N_{g})\phi^{+}_{k}(M_{l}\otimes N_{g})^{\dagger})\notag\\
    &\qquad\qquad\times G_{k}(\phi_{ij}^{AB}).
    \label{eq::gkinkappaprooof2}
\end{align}
Using Eqs. \eqref{eq::kappa}, \eqref{eq::gkinkappaprooof1} and \eqref{eq::gkinkappaprooof2} we get
\begin{align}
    &\sum_{ijgl}q_j\,r_{lgi}\,G_{k}(\phi_{ijgl}^{AB})
    \leq \sum_{ijh}q_jG_{k}((K_{i}\otimes\mathbb{I})\phi^{+}_{k}(K_{i}\otimes\mathbb{I})^{\dagger})\,
     G_{k}(\phi_{ij}^{AB})\notag\\
    &\leq \kappa_s\left(\sum_{ij}q_iG_{k}(\phi_{ij}^{AB})\right)\leq \kappa_s.
\end{align}
 In the first inequality we used that $G_{k}((M_{l}\otimes N_{g})\phi^{+}_{k}(M_{l}\otimes N_{g})^{\dagger})\leq 1.$

We have showed that our conjecture holds for $n=1.$ Let us assume it holds for some $n.$ Let $\Lambda_n(\rho^{AB})=\sum_i s_i\Xi_i^{AB},$ where each $\Xi_i^{AB}$ has Schmidt rank at most $k$ because of the induction assumption.
 We can write:
\begin{align}
    \Lambda_{n+1}(\rho^{AB})=&\Lambda \circ F_{n+1}\circ \Lambda_{n}(\rho^{AB})\nonumber\\ &= \Lambda \circ F_{n+1}(\sum_i s_i\Xi_i^{AB})=
    \sum_{ighl} s_ir_{ghl}\phi_{ihgl}^{AB},
\end{align}
where
\begin{align}
    \ket{\phi_{hikl}}^{AB} = \frac{1}{\sqrt{r_{h}r'_{lg}}} (K_{h}\otimes\mathbb{I})(M_{l}\otimes N_{g})\,\ket{\Xi_{i}}^{AB},\\
    r_{hgl}=\operatorname{Tr}[(K_{h}\otimes\mathbb{I})(M_{l}\otimes N_{g})\,\Xi_{i}^{AB}(M_{l}\otimes N_{k})^{\dagger}(K_{h}\otimes\mathbb{I})^{\dagger}],
\end{align}
where each $\ket{\phi_{hikl}}^{AB}$ has Schmidt rank not greater than $k.$ This completes the first part of the Proposition.
Using the induction assumption and Eq.~\eqref{eq::kappa} , we obtain:
\begin{align}
    &\sum_{ighl} s_ir_{ghl}G_k(\phi_{ihgl}^{AB})=\notag\\
    &\sum_{ighl} s_iG_{k}((K_{h}\otimes\mathbb{I})\phi^{+}_{k}(K_{h}\otimes\mathbb{I})^{\dagger})\,G_{k}((M_{l}\otimes N_{g})\phi^{+}_{k}(M_{l}\otimes N_{g})^{\dagger})\notag\\
    &\times G_{k}(\Xi_{i}^{AB})\leq \sum_{h} G_{k}((K_{h}\otimes\mathbb{I})\phi^{+}_{k}(K_{h}\otimes\mathbb{I})^{\dagger})(\sum s_iG_k( \Xi_i^{AB}))\notag \\
    &\leq \kappa_s^{n+1}. 
\end{align}
In the last step, we used the induction assumption and Eq.~\eqref{eq::kappa}. It completes the proof. 
\end{proof}

Now, let us define the $k$-Schmidt fidelity of a bipartite pure state $\ket{\psi}^{AB}$ as:
\begin{equation}\label{def:k_schmid_fid}
    F_{k}(\psi^{AB})=\max_{\phi\in\mathcal{S}_{k}}F(\psi^{AB},\phi^{AB}).
\end{equation}
where $\mathcal{S}_{k}$ denotes the set of states that can be expressed as convex combinations of bipartite pure states with Schmidt rank at most $k$. Since $\ket{\psi^{AB}}$
 is pure, the maximum in \eqref{def:k_schmid_fid} is attained by a pure state.
\begin{prop}{\label{l::fidk}}
For a bipartite pure state $\ket{\psi}^{AB}$, the following holds:
\begin{equation}
F_{k}(\psi^{AB})=\sum_{i=0}^{k-1}\lambda_{i},
\end{equation}
where $\lambda_i$ are Schmidt coefficients of $\ket{\psi}$ in decreasing order. Additionally, the state $\ket{\phi}^{AB}\in\mathcal{S}_{k}$ that achieves the maximum in Eq. \eqref{def:k_schmid_fid} is given by
\begin{equation}
\ket{\phi}^{AB}=\frac{1}{\sqrt{\sum_{j=0}^{k-1}\lambda_{j}}}\sum_{i=0}^{k-1}\sqrt{\lambda_{i}}\ket{ii}^{AB}
\label{eq::optimalphi}
\end{equation}
where $\ket{i}$ are the Schmidt vectors of $\ket{\psi}$.
\end{prop}

\begin{proof}
Let $\ket{\psi}=\sum_{i,j} M_{ij}\ket{i}\ket{j}$ and $\ket{\phi}=\sum_{i,j}N_{ij}\ket{i}\ket{j}.$ Their inner product can be written as
    \begin{equation}
        \vert\!\braket{\phi|\psi}\!\vert=\left\vert\sum_{i,j}N_{ij}^{*}M_{ij}\right\vert=\left\vert\mathrm{Tr}\left(N^{\dagger}M\right)\right\vert.
        \label{eq::innertraceequality}
    \end{equation}
    Using the von Neumann trace inequality~\cite{Horn_Johnson_1991}, we obtain: 
    \begin{equation}
        \left\vert\mathrm{Tr}\left(N^{\dagger}M\right)\right\vert\leq \sum_{i=0}^{d-1}s_i(N)s_i(M)=\sum_{i=0}^{k-1}s_i(N)\sqrt{\lambda_i},
        \label{eq::vonNeumanninequality}
    \end{equation}
    where $s_i(X)$ denote the singular values of $X$ in non-increasing order. In the last step, we used that the singular values of $M$ are the Schmidt coefficients of $\ket{\psi}$.

Since $\ket{\phi}$ has Schmidt rank $k$, the matrix $N$ has rank at most $k$, and hence only the first $k$ singular values are non-zero. Applying the Cauchy--Schwarz inequality, we obtain
    \begin{equation}
        \sum_{i=0}^{k-1}s_i(N)\sqrt{\lambda_i}\leq \sqrt{\sum_{i=0}^{k-1}s_i(N)^{2}}\sqrt{\sum_{i=0}^{k-1}\lambda_i}=\sqrt{\sum_{i=0}^{k-1}\lambda_i}.
        \label{eq::CaScinequalityapplied}
    \end{equation}
    In the last step we used that squares of singular values of $N$ must add up to 1. Thus, we have the inequality $ \vert\!\braket{\phi | \psi}\!\vert\leq \sqrt{\sum_{i=0}^{k-1}\lambda_i}.$ Taking squares we obtain:
    \begin{equation}
         \vert\!\braket{\phi | \psi}\!\vert^{2}\leq \sum_{i=0}^{k-1}\lambda_i.
         \label{eq::Fp}
    \end{equation}
    We have thus established an upper bound on the fidelity. One can verify that the bound is saturated by the state given in Eq.~\eqref{eq::optimalphi}, which completes the proof.   
\end{proof}

Analogous to the definition of $k$-Schmidt fidelity for pure states, we define $k$-Schmidt fidelity for a bipartite mixed state $\rho^{AB}$ as follows:
\begin{equation}
    F_{k}(\rho^{AB})=\max_{\widetilde{\sigma}^{AB}\in\mathcal{S}_{k}}F(\rho^{AB},\widetilde{\sigma}^{AB}).
    \label{eq::definingfidelitymixed}
\end{equation}
\begin{prop}{\label{le::mixfik}}
    For any pure-state decomposition $\{p_{i},\psi_i^{AB}\}$
    of $\rho^{AB}$ it holds that 
    \begin{equation}
    F_{k}(\rho^{AB})\geq\sum_{i}p_{i}F_{k}(\psi_{i}^{AB}).
    \end{equation}
\end{prop}
\begin{proof}
For every pure state $\ket{\psi_i^{AB}}$ appearing in a decomposition of $\rho$, there exists a Schmidt-rank-at-most-$k$ state $\ket{\phi_i^{AB}}$ that optimally achieves the fidelity, i.e.,
\begin{equation}
    F_k(\psi_i^{AB})=\lvert \braket{\psi_i^{AB}\vert \phi_i^{AB}}\rvert^{2}.
\end{equation}
We can always set phases such that inner product $\braket{\psi_i^{AB}\vert \phi_i^{AB}}$ is non-negative. We obtain:
\begin{equation}
    \sqrt{F_k(\psi_i^{AB})}=\braket{\psi_i^{AB}\vert \phi_i^{AB}}.
    \label{eq::fidsqu}
\end{equation}
Let us define new probabilities $q_i$ such that:
\begin{equation}
    q_i=\frac{p_i F_k(\psi_i^{AB})}{T},
    \label{eq::qprob}
\end{equation}
where $T=\sum_j p_j F_k(\psi_j^{AB})$ is a normalization constant. Let us define:
\begin{equation}
    \sigma^{AB}=\sum_i q_i\ket{\phi_i^{AB}}\bra{\phi_i^{AB}}.
    \label{eq::sigmamixedfidelity}
\end{equation}
Since every state $\ket{\phi_i^{AB}}$ has Schmidt rank at most $k,$ $\sigma^{AB}\in S_k.$ 

Let us purify $\sigma^{AB}$ and $\rho^{AB}.$ We obtain:
\begin{equation}
    \begin{split}
        &\ket{\Psi^{\rho}}=\sum_i \sqrt{p_i}\ket{i}^{R}\ket{\psi_i}^{AB}\\
        &\ket{\Phi^{\sigma}}=\sum_i \sqrt{q_i}\ket{i}^{R}\ket{\phi_i}^{AB}.
    \end{split}
\end{equation}

By Uhlmann's theorem we obtain~\cite{Uhlmannfidelity}
\begin{equation}
    F(\rho^{AB},\sigma^{AB})=\max_{U^{R}}\lvert \bra{\Phi^{\sigma}}U^{R}\otimes \openone^{AB}\ket{\Psi^{\rho}}\rvert^{2},
\end{equation}
where maximization is taken over all unitaries $U^{R}$ on purifying system. When we take specific $U^{R}=\openone^{R}$ we obtain:
\begin{equation}
    F(\rho^{AB},\sigma^{AB})\geq \lvert \braket{\Phi^{\sigma}\vert \Psi^{\rho}}\rvert^{2}.
    \label{eq::fidineq}
\end{equation}
We can further write:
\begin{equation}
     \braket{\Phi^{\sigma}\vert \Psi^{\rho}}=\sum_i \sqrt{p_iq_i}\braket{\psi_i^{AB}\vert \phi_i^{AB}}
\end{equation}
Using Eq.~\eqref{eq::fidsqu} and \eqref{eq::qprob} we obtain:
\begin{equation}
    \begin{split}
        &\braket{\Phi^{\sigma}\vert \Psi^{\rho}}=\sum_i \sqrt{p_i\frac{p_iF_k(\psi_i^{AB})}{T}}\sqrt{F_k(\psi_i^{AB})}=\\
        &\frac{\sum_i p_i F_k(\psi_i^{AB})}{\sqrt{\sum_i p_i F_k(\psi_i^{AB})}}=
        \sqrt{\sum_i p_i F_k(\psi_i^{AB})}.
    \end{split}
    \label{eq::purificationexpansion}
\end{equation}
Using relations \eqref{eq::fidineq} and \eqref{eq::purificationexpansion} we obtain:
\begin{equation}
    F(\rho^{AB},\sigma^{AB})\geq \sum_i p_i F_k(\psi_i^{AB})
\end{equation}
Since in Eq.~\eqref{eq::sigmamixedfidelity} we constructed specific $\sigma^{AB}$ there is a relation:
\begin{equation}
    \max_{\widetilde{\sigma}^{AB}\in\mathcal{S}_{k}}F(\rho^{AB},\widetilde{\sigma}^{AB})\geq F(\rho^{AB},\sigma^{AB})\geq \sum_i p_i F_k(\psi_i^{AB}),
\end{equation}
which completes the proof. Similar technique was used in~\cite{Philipschrodingerasaquantumprogrammer}.
\end{proof}

For a pure state $\ket{\psi}$ with Schmidt rank at most $k$, we
further have 
\begin{equation}
\lambda_{k-1}\leq\left(\prod_{i=0}^{k-1}\lambda_{i}\right)^{1/k}=\frac{G_{k}(\psi)}{k}.
\end{equation}
Using Proposition~\ref{l::fidk} we arrive at
\begin{equation}
F_{k-1}(\psi)=1-\lambda_{k-1}\geq1-\frac{G_{k}(\psi)}{k}.
\end{equation}

Using now Propositions~\ref{pr::expon} and~\ref{le::mixfik}, it follows that the state
$\Lambda_{n}[\rho^{AB}]$ can be written as $\Lambda_{n}[\rho^{AB}]=\sum_{i}s_{i}\phi_{i}^{AB}$
with probabilities $s_{i}$ and pure states $\phi_{i}^{AB}$ having
Schmidt rank at most $k$ such that 
\begin{align}
F_{k-1}\left(\Lambda_{n}[\rho^{AB}]\right) & \geq\sum_{i}s_{i}F_{k-1}\left(\phi_{i}^{AB}\right)\\
 & \geq1-\frac{1}{k}\sum_{i}s_{i}G_{k}\left(\phi_{i}^{AB}\right)\nonumber \\
 & \geq1-\frac{\kappa_s^{n}}{k}.\nonumber 
\end{align}
These arguments together with the Fuchs--van de Graaf inequality~\cite{Fuchsinequality} imply that we can bound the trace-distance to the set $\mathcal S_{k-1}$ as follows:
\begin{align} \label{eq:Bound-1}
\min_{\sigma^{AB}\in\mathcal{S}_{k-1}}\Vert\Lambda_{n}[\rho^{AB}]-\sigma^{AB}\Vert_{1} & \leq2\sqrt{1-F_{k-1}\left(\Lambda_{n}[\rho^{AB}]\right)}\nonumber \\
 & \leq2\frac{\kappa_s^{n/2}}{\sqrt{k}}.
\end{align}

Note now that the overall protocol $\Lambda_n$ can always be decomposed into two protocols $\Lambda_{n_2}\circ \Lambda_{n_1}$ with $n=n_1+n_2$. We now define the state $\nu_{k-1}^{AB}\in\mathcal{S}_{k-1}$ to be the closest state to $\Lambda_{n_{1}}[\rho^{AB}]$. We further define $\nu_{k-2}\in\mathcal{S}_{k-2}$ to be the closest state to $\Lambda_{n_{2}}[\nu_{k-1}^{AB}]$. Using triangle inequality and data processing inequality we find 
\begin{align}
\Vert\Lambda_{n_{2}}\circ\Lambda_{n_{1}}[\rho^{AB}]-\nu_{k-2}^{AB}\Vert_{1} & \leq\Vert\Lambda_{n_{2}}\circ\Lambda_{n_{1}}[\rho^{AB}]-\Lambda_{n_{2}}[\nu_{k-1}^{AB}]\Vert_{1}\nonumber \\
 & +\Vert\Lambda_{n_{2}}[\nu_{k-1}^{AB}]-\nu_{k-2}^{AB}\Vert_{1}\nonumber \\
 & \leq\Vert\Lambda_{n_{1}}[\rho^{AB}]-\nu_{k-1}^{AB}\Vert_{1}\nonumber \\
 & +\Vert\Lambda_{n_{2}}[\nu_{k-1}^{AB}]-\nu_{k-2}^{AB}\Vert_{1}\nonumber \\
 & \leq2\left(\frac{\kappa_s^{n_{1}/2}}{\sqrt{k}}+\frac{\kappa_s^{n_{2}/2}}{\sqrt{k-1}}\right).
\end{align}
Note that we can decompose the overall protocol $\Lambda_n$ into $k-1$ smaller protocols: $\Lambda_{n_{k-1}}\circ\Lambda_{n_{k-2}} \circ\cdots\circ \Lambda_{n_1}$ with $n=\sum_{i=1}^{k-1}n_{i}$. As before, we define the state $\nu_{k-1}^{AB}\in\mathcal{S}_{k-1}$ to be the closest state to $\Lambda_{n_{1}}[\rho^{AB}]$. We then define an analogous state for each $n_i$ in the following fashion: let $\nu^{AB}_{k-i}\in\mathcal{S}_{k-i}$ be the state closest to $\Lambda_{n_{i}}[\nu^{AB}_{k-i+1}]$. We obtain
\begin{align}
&\Vert\Lambda_{n_{k-1}}\circ\cdots\circ\Lambda_{n_{2}}\circ\Lambda_{n_{1}}[\rho^{AB}]-\nu^{AB}_{1}\Vert_{1}\nonumber\\
&\leq \Vert\Lambda_{n_{k-1}}\circ\cdots\circ\Lambda_{n_{2}}\circ\Lambda_{n_{1}}[\rho^{AB}]-\Lambda_{n_{k-1}}\circ\cdots\circ\Lambda_{n_{2}}[\nu^{AB}_{k-1}]\Vert_{1}\nonumber\\
&\quad + \Vert\Lambda_{n_{k-1}}\circ\cdots\circ\Lambda_{n_{2}}[\nu^{AB}_{k-1}]-\nu^{AB}_{1}\Vert_{1}.
\end{align}
Then using data processing~\cite{Nielsen_Chuang_2010}, we get
\begin{align}
&\Vert\Lambda_{n_{k-1}}\circ\cdots\circ\Lambda_{n_{2}}\circ\Lambda_{n_{1}}[\rho^{AB}]-\nu^{AB}_{1}\Vert_{1}\nonumber\\
&\leq \Vert\Lambda_{n_{1}}[\rho^{AB}]-[\nu^{AB}_{k-1}]\Vert_{1}\nonumber\\
&\qquad + \Vert\Lambda_{n_{k-1}}\circ\cdots\circ\Lambda_{n_{2}}[\nu^{AB}_{k-1}]-\nu^{AB}_{1}\Vert_{1}.
\end{align}
Using the triangle inequality once more, we get
\begin{align}
&\Vert\Lambda_{n_{k-1}}\circ\cdots\circ\Lambda_{n_{2}}\circ\Lambda_{n_{1}}[\rho^{AB}]-\nu^{AB}_{1}\Vert_{1}\nonumber\\
&\leq \Vert\Lambda_{n_{1}}[\rho^{AB}]-[\nu^{AB}_{k-1}]\Vert_{1}\nonumber\\
&+ \Vert\Lambda_{n_{k-1}}\circ\cdots\circ\Lambda_{n_{3}}\circ\Lambda_{n_{2}}[\rho^{AB}]-\Lambda_{n_{k-1}}\circ\cdots\circ\Lambda_{n_{3}}[\nu^{AB}_{k-2}]\Vert_{1}\nonumber\\
&\quad + \Vert\Lambda_{n_{k-1}}\circ\cdots\circ\Lambda_{n_{3}}[\nu^{AB}_{k-2}]-\nu^{AB}_{1}\Vert_{1}.
\end{align}
Then using data processing, we arrive at
\begin{align}
&\Vert\Lambda_{n_{k-1}}\circ\cdots\circ\Lambda_{n_{2}}\circ\Lambda_{n_{1}}[\rho^{AB}]-\nu^{AB}_{1}\Vert_{1}\nonumber\\
&\leq \Vert\Lambda_{n_{1}}[\rho^{AB}]-[\nu^{AB}_{k-1}]\Vert_{1}+\Vert\Lambda_{n_{2}}[\nu^{AB}_{k-1}]-[\nu^{AB}_{k-2}]\Vert_{1}\nonumber\\
&\qquad + \Vert\Lambda_{n_{k-1}}\circ\cdots\circ\Lambda_{n_{3}}[\nu^{AB}_{k-2}]-\nu^{AB}_{1}\Vert_{1}.
\end{align}
Iterating this procedure, we obtain the inequality
\begin{align}
&\Vert\Lambda_{n_{k-1}}\circ\cdots\circ\Lambda_{n_{2}}\circ\Lambda_{n_{1}}[\rho^{AB}]-\nu^{AB}_{1}\Vert_{1}\nonumber\\
&\leq \Vert\Lambda_{n_{1}}[\rho^{AB}]-\nu^{AB}_{k-1}\Vert_{1}+ \sum_{i=2}^{k-1}\Vert\Lambda_{n_{i}}[\nu^{AB}_{k-i+1}]-\nu^{AB}_{k-i}\Vert_{1}.
\end{align}
Using \eqref{eq:Bound-1}, we get
\begin{align}
\min_{\sigma^{AB}\in\mathcal{S}_{1}}\Vert\Lambda_{n_{k-1}}\circ\cdots\circ\Lambda_{n_{2}}\circ\Lambda_{n_{1}}[\rho^{AB}]-\sigma^{AB}\Vert_{1} & \leq2\sum_{i=1}^{k-1}\frac{\kappa_s^{n_{i}/2}}{\sqrt{k-i+1}}.
\end{align}
Choosing $n_i = N$ and $n = N(k-1)$ we obtain 
\begin{align}
\sum^{k-1}_{i=1}\frac{\kappa_s^{n_{i}/2}}{\sqrt{k-i+1}} & =\kappa_s^{N/2}\sum^{k-1}_{i=1}\frac{1}{\sqrt{k-i+1}}\leq2\kappa_s^{N/2}(\sqrt{k}-1)\nonumber \\
 & =2\kappa_s^{n/[2(k-1)]}(\sqrt{k}-1)\leq2\kappa_s^{n/[2(d_{A}-1)]}(\sqrt{d_{A}}-1). \label{eq::almostseparability}
\end{align}

Collecting the above arguments, we thus have
\begin{equation}
    \min_{\sigma^{AB}\in\mathcal{S}_{1}}\Vert\Lambda_{n}[\rho^{AB}]-\sigma^{AB}\Vert_{1}\leq 4\kappa_s^{n/[2(k-1)]}(\sqrt{k}-1).
    \label{eq::separabilityk}
\end{equation}
Noting that $k\leq d_A,$ we obtain
\begin{equation}
\min_{\sigma^{AB}\in\mathcal{S}_{1}}\Vert\Lambda_{n}[\rho^{AB}]-\sigma^{AB}\Vert_{1}\leq 4\kappa_s^{n/[2(d_{A}-1)]}(\sqrt{d_{A}}-1).
\label{eq::separability}
\end{equation}
Taking the optimal Kraus representation of $\Lambda$ we obtain:
\begin{equation}
    \min_{\sigma^{AB}\in\mathcal{S}_{1}}\Vert\Lambda_{n}[\rho^{AB}]-\sigma^{AB}\Vert_{1}\leq 4\kappa^{n/[2(d_{A}-1)]}(\sqrt{d_{A}}-1).
\end{equation}
This completes the proof of the theorem.

\subsection*{Proof of Theorem \ref{th::scaling}.}
\textbf{First part:} Assume that the channel $\Lambda$ admits the decomposition
\begin{equation}
    \Lambda
    =
    p\mathcal{E}
    +
    (1-p)\mathcal{M},
    \qquad
    0<p<1,
\end{equation}
where $\mathcal{E}$ is an entanglement-breaking channel and
$\mathcal{M}$ is an arbitrary quantum channel.

We will show that, for $m$ parallel uses of $\Lambda$, the quantity
$\kappa$ introduced in Eq. \eqref{eq::separabilityk} satisfies
\begin{equation}
    \kappa
    \leq
    1-p^m
\end{equation}
for every $k\geq 2$. The bound is uniform with respect to the
choice of the $k$-dimensional Schmidt support.

Every entanglement-breaking channel admits a Kraus
representation consisting of rank-one operators\cite{HorodeckiEB}. We may therefore
write
\begin{equation}
    \mathcal{E}(\rho)
    =
    \sum_a A_a \rho A_a^\dagger,
    \qquad
    \operatorname{rank}(A_a)=1,
    \qquad
    \sum_a A_a^\dagger A_a=\openone.
\end{equation}
Let
\begin{equation}
    \mathcal{M}(\rho)
    =
    \sum_b B_b \rho B_b^\dagger,
    \qquad
    \sum_b B_b^\dagger B_b=\openone
\end{equation}
be an arbitrary Kraus representation of $\mathcal{M}$.
Consequently, a Kraus representation of $\Lambda$ is given by
\begin{equation}
    \left\{\sqrt{p}\,A_a\right\}_a
    \cup
    \left\{\sqrt{1-p}\,B_b\right\}_b.
\end{equation}

Consider now the channel $\Lambda^{\otimes m}$. Its Kraus
operators are tensor products of the single-system Kraus
operators. In particular, the subset of Kraus operators for which
all $m$ factors originate from the entanglement-breaking part is
given by
\begin{equation}
    R_{\boldsymbol{a}}
    =
    p^{m/2}
    A_{a_1}\otimes\cdots\otimes A_{a_m},
    \qquad
    \boldsymbol{a}=(a_1,\ldots,a_m).
\end{equation}
Since every $A_{a_i}$ has rank one, we have
\begin{equation}
    \operatorname{rank}(R_{\boldsymbol{a}})=1.
\end{equation}

Let $\ket{\phi_k^+}$ be a maximally entangled state of Schmidt
rank $k$, whose Schmidt support on the transmitted subsystem is an
arbitrary $k$-dimensional subspace. For every
$\boldsymbol{a}$, the state
\begin{equation}
    (R_{\boldsymbol{a}}\otimes\openone)
    \ket{\phi_k^+}
\end{equation}
has Schmidt rank at most one. Hence, for every $k\geq 2$,
\begin{equation}
    G_k\!\left[
        (R_{\boldsymbol{a}}\otimes\openone)
        \phi_k^+
        (R_{\boldsymbol{a}}^\dagger\otimes\openone)
    \right]
    =
    0.
\end{equation}

Let
\begin{equation}
    q_{\boldsymbol{a}}
    =
    \operatorname{Tr}\!\left[
        (R_{\boldsymbol{a}}\otimes\openone)
        \phi_k^+
        (R_{\boldsymbol{a}}^\dagger\otimes\openone)
    \right]
\end{equation}
denote the probability of the corresponding Kraus outcome.
The total probability of all outcomes containing only Kraus
operators from the entanglement-breaking part is
\begin{align}
    \sum_{\boldsymbol{a}}q_{\boldsymbol{a}}
    &=
    \operatorname{Tr}\!\left[
        \left(
            \sum_{\boldsymbol{a}}
            R_{\boldsymbol{a}}^\dagger
            R_{\boldsymbol{a}}
            \otimes\openone
        \right)
        \phi_k^+
    \right]
    \\
    &=
    \operatorname{Tr}\!\left[
        \left(
            p^m
            \left(
                \sum_a A_a^\dagger A_a
            \right)^{\otimes m}
            \otimes\openone
        \right)
        \phi_k^+
    \right]
    \\
    &=
    p^m.
\end{align}

Let $L_a$ be a fixed Kraus operators of $\Lambda^{\otimes m}.$ Since $L_a$ is a fixed Kraus representation, we obtain:
\begin{equation}
    \kappa\leq \sum_a G_k(L_{a}\otimes\openone)
        \phi_k^+
        (L_{a}^\dagger\otimes\openone))
        \label{eq::smkappaleb}
\end{equation}
We now evaluate the average $G_k$-concurrence associated with
this particular Kraus decomposition. Let
$\mathcal{I}_{\mathrm{EB}}$ denote the set of outcomes in which
all factors originate from $\mathcal{E}$, and let
$\mathcal{I}_{\mathrm{rest}}$ denote the remaining outcomes.
For each outcome $\alpha$, let $q_\alpha$ be its probability
and let $\theta_\alpha$ be the corresponding normalized state.
Since
\begin{equation}
    0\leq G_k(\theta_\alpha)\leq 1,
\end{equation}
we obtain
\begin{align}
   \kappa \leq &\sum_\alpha q_\alpha G_k(\theta_\alpha)
    =
    \sum_{\alpha\in\mathcal{I}_{\mathrm{EB}}}
        q_\alpha G_k(\theta_\alpha)
    +
    \sum_{\alpha\in\mathcal{I}_{\mathrm{rest}}}
        q_\alpha G_k(\theta_\alpha)
    \nonumber\\
    &=
    \sum_{\alpha\in\mathcal{I}_{\mathrm{rest}}}
        q_\alpha G_k(\theta_\alpha)
    \leq
    \sum_{\alpha\in\mathcal{I}_{\mathrm{rest}}}
        q_\alpha=1-p^m.
\end{align}
Introducing
\begin{equation}
    \gamma=-\ln p>0,
\end{equation}
we may equivalently write
\begin{equation}
        \kappa
        \leq
        1-e^{-\gamma m}
\end{equation}

On the other hand, Eq. \eqref{eq::separabilityk}, applied to the channel
$\Lambda^{\otimes m}$, gives
\begin{equation}
    \min_{\sigma^{AB}\in\mathcal{S}}
    \left\|
        \rho_{n,m(n)}^{AB}-\sigma^{AB}
    \right\|_1
    \leq
    4\kappa^{n/(2(k_m-1))}
    \left(\sqrt{k_m}-1\right),
\end{equation}
where $k_m$ is the relevant Schmidt-rank parameter. 

It is possible to establish entanglement using a channel if
\begin{equation}
    \liminf_{n\rightarrow\infty}
    E_f\!\left(\rho_{n,m(n)}^{AB}\right)
    >0.
\end{equation}
Using lower bound on distance to separable states from Lemma \ref{le::epsilondistance} we obtain:
\begin{equation}
    \frac{E_0^2}
    {8m^2(\log_2d)^2}
    \leq
    2
    \left(1-e^{-\gamma m}\right)^{
        \frac{n}{2(k_m-1)}
    }
    \left(\sqrt{k_m}-1\right).
\end{equation}
Equivalently,
\begin{equation}
    \left(1-e^{-\gamma m}\right)^{
        \frac{n}{2(k_m-1)}
    }
    \geq
    \frac{E_0^2}{
        16m^2(\log_2d)^2
        \left(\sqrt{k_m}-1\right)
    }.
\end{equation}
Taking the natural logarithm gives
\begin{equation}
    \frac{n}{2(k_m-1)}
    \ln\!\left(1-e^{-\gamma m}\right)
    \geq
    -
    \ln\!\left[
        \frac{
            16m^2(\log_2d)^2
            \left(\sqrt{k_m}-1\right)
        }{E_0^2}
    \right].
\end{equation}
Using
\begin{equation}
    \ln(1-x)\leq-x,
    \qquad 0<x<1,
\end{equation}
we arrive at
\begin{equation}
    \frac{ne^{-\gamma m}}{2(k_m-1)}
    \leq
    \ln\!\left[
        \frac{
            16m^2(\log_2d)^2
            \left(\sqrt{k_m}-1\right)
        }{E_0^2}
    \right].
\end{equation}
Consequently,
\begin{equation}
    n
    \leq
    2(k_m-1)e^{\gamma m}
    \ln\!\left[
        \frac{
            16m^2(\log_2d)^2
            \left(\sqrt{k_m}-1\right)
        }{E_0^2}
    \right].
\end{equation}

For a block of $m$ qudits,
\begin{equation}
    k_m\leq d^m.
\end{equation}
It follows that
\begin{align}
    n
    &\leq
    2d^m e^{\gamma m}
    \left[
        \frac{m}{2}\ln d
        +
        2\ln m
        +
        C_0
    \right]\nonumber
    \\
    &\leq
    C_1 m
    e^{(\gamma+\ln d)m}
\end{align}
for all sufficiently large $m$, where $C_0$ and $C_1$ are
positive constants independent of $n$ and $m$. We used the fact that $m+\ln m<Cm$ for sufficiently large $m.$

Therefore,
\begin{equation}
    m
    \geq
    \frac{1}{\gamma+\ln d}
    W_0\!\left(
        \frac{\gamma+\ln d}{C_1}n
    \right),
\end{equation}
where $W_0$ is the principal branch of the Lambert function. Using\cite{CorlessLambert}
\begin{equation}
    W_0(x)
=
\ln x-\ln\ln x
+
O\left(
\frac{\ln\ln x}{\ln x}
\right),
\label{eq::w0}
\end{equation}
we obtain
\begin{equation}
    m
    \geq
    \frac{\ln n-\ln\ln n}
         {\gamma+\ln d}
    -
    O(1)
\end{equation}
Thus, preserving a nonvanishing amount of entanglement of
formation requires at least a logarithmic number of parallel
channel uses per link.

\textbf{Second part:}

In the beginning we define a notion that will be used in the proof. Let $\mathcal{P}$ be a projector into a Schmidt support of $\ket{\phi_k^{+}}^{AB}=\sum_i \sqrt{\frac{1}{k}}\ket{e_i}^{A}\ket{f_i}^{B}.$ That is $\mathcal{P}$ projects onto $\operatorname{span}(\ket{e_1},\cdots ,\ket{e_k}).$ Let $\mathcal{V}$ be isometry $\mathbb{C}^{k}\rightarrow \mathcal{H_A}$ defined by $\mathcal{V}\ket{j}^{K}=\ket{e_j}^{A},$ where superscripts highlight that $\ket{j}^{K}\in\mathbb{C}^{k}$ and $\ket{e}_j^{A}\in\mathcal{H}_A.$ Operator $\mathcal{V}^{\dagger}X\mathcal{V}$ is a $k\times k$ matrix on $\mathcal{C}^{k}$ with matrix elements $(\mathcal{V}^{\dagger}X\mathcal{V})_{ij}=\bra{e_i}X\ket{e_j}.$ We define $\det_{\mathcal{P}}X=det(V^{\dagger}XV).$ Consequently, $\det_{\mathcal{P}}X$ is the product of the eigenvalues of the compression $\mathcal P X\mathcal P$, regarded as an operator on the Schmidt support $\mathcal P\mathcal H_A$.  

Let us consider a qudit amplitude damping channel $\Lambda_{AD}^{\otimes m}.$ Kraus operators of a single qudit amplitude damping channel $\Lambda_{AD}$ are:
\begin{equation}
    \begin{split}
        &K_0=\sum_{j=0}^{d-1} \sqrt{q_j}\ket{j}\bra{j}\\
        &K_{i\leftarrow j}=\sqrt{\gamma_{i\leftarrow j}}\ket{i}\bra{j},
    \end{split}
\end{equation}
where $\eta_j =\sum_{i<j}\gamma_{i\leftarrow  j}$ is total decay probability of the $j$-th level and $q_j=1-\eta_j$ is no decay probability. This kind of amplitude damping channel was studied in \cite{Chessaad}.

Assume that we have a Hilbert space with bipartition $A|B=A_1\cdots A_m|B.$ Let us consider states with Schmidt rank at most $k.$ Let us define a maximally entangled state on this Hilbert space $\ket{\phi_k^{+}}=\frac{1}{\sqrt{k}}\sum_{i=0}^{k-1}\ket{e}_i\ket{f_i}.$ Let $\mathcal{P}$ be the projector onto the Schmidt support of $\ket{\phi_k^{+}},$ i.e. space spanned by $\operatorname{span}(\ket{e}_1,\cdots \ket{e}_k).$ 

To upper bound $\kappa$ we need to find:
\begin{equation}
    \kappa \leq \sum_iG_k\left( L_i\otimes \openone \phi_k^{+}L_i^{\dagger}\otimes \openone \right),
\end{equation}
where $L_i=K_{\alpha_1}\otimes \cdots \otimes K_{\alpha_m}$ is the standard Kraus representation of $\Lambda_{AD}^{\otimes m}.$
Let $E_l=L_l^{\dagger}L_l.$ We are interested only by the action of $L_l$ on the Schmidt support. We can compute it by $L_l\mathcal{P}.$ We have $(L_l\mathcal{P})^{\dagger}L_l\mathcal{P}=\mathcal{P}E_l\mathcal{P}$ As discussed in Lemma~\ref{lem:Gk}, we can express $G_k$-concurrence of a single Kraus branch as: 
\begin{equation}
    G_k(L_l\otimes \openone \phi_k^{+}L_l^{\dagger}\otimes \openone)=\left(\prod_{i=1}^{k}s_i^{2}(L_l\mathcal{P})\right)^{1/k}=\det_{\mathcal{P}} (\mathcal{P}E_l\mathcal{P})^{1/k},
\end{equation}
where $s_i(L_lP)$ are singular values of $L_l\mathcal{P}.$ For fixed $\mathcal{P}$ we have:
\begin{equation}
    \kappa \leq \sum_l \det_{\mathcal{P}}(\mathcal{P}E_l\mathcal{P})^{1/k}.
    \label{eq::kappadet}
\end{equation}
We define:
\begin{equation}
    X_l=\mathcal{P}E_l\mathcal{P}.
\end{equation}
Since $\sum_l E_l=\openone,$ we obtain 
\begin{equation}
    \sum_l X_l=\mathcal{P}
\end{equation}
Now we proceed to find a bound on Eq. \eqref{eq::kappadet}. We do it using the eigenvalues $\{ x_{l,j} \}_{j=1}^{k}$ of the operator $X_l$. The arithmetic mean of $x_{l,j}$ is $a_l=\frac{1}{k} \mathrm{Tr}(X_l),$, while their geometric mean is $\det_{\mathcal{P}}(X_l)^{1/k}$. Using Cartwright-Field inequality \cite{Cartwrightinequality}, we know that:
\begin{equation}
    a_l-\det_{\mathcal{P}}(X_l)^{1/k}\geq \frac{1}{2kR_l}\sum_{j=1}^{k}(x_{l,j}-a_l)^{2},
\end{equation}
where $R_l=\max \{  x_{l,j}\}_{j=1}^{k}.$ Eigenvalues of $X_l$ are bounded by 1. Therefore, we obtain a bound $R_l\leq 1.$

On the subspace $\mathcal{P}\mathcal{H}_A$, the projector
$\mathcal{P}$ acts as the identity. Therefore,
\begin{equation}
    a_l-
    \left[
        \det\nolimits_{\mathcal{P}}(X_l)
    \right]^{1/k}
    \geq
    \frac{1}{2k}
    \left\|
        X_l-a_l\mathcal{P}
    \right\|_2^2.
\end{equation}
Summing over all Kraus outcomes $l$, we obtain
\begin{equation}
    \sum_l \left(a_l-\left[\det\nolimits_{\mathcal{P}}(X_l)\right]^{1/k}\right)\geq
    \frac{1}{2k}\sum_l\left\| X_l-a_l\mathcal{P}\right\|_2^2.
\end{equation} 
Recalling that
\begin{equation}
    X_l=\mathcal{P}E_l\mathcal{P},
    \qquad
    a_l=
    \frac{\operatorname{Tr}(\mathcal{P}E_l\mathcal{P})}{k},
\end{equation}
we arrive at
\begin{equation}
    1-\sum_ldet_{\mathcal{P}}(\mathcal{P}E_l\mathcal{P})^{1/k}
    \geq
    \frac{1}{2k}
    \sum_l
    \left\|
        \mathcal{P}E_l\mathcal{P}
        -
        \frac{
            \operatorname{Tr}(\mathcal{P}E_l\mathcal{P})
        }{k}
        \mathcal{P}
    \right\|_2^2.
    \label{eq::smtobound}
\end{equation}

We need to derive a lower bound on the right-hand side of \eqref{eq::smtobound}. Let us consider the single channel $\Lambda_{AD}.$ Let $D_j=\ket{j}\bra{j}.$ We obtain $E_0^{s}=K_0^{\dagger}K_0=\sum_{j=0}^{d-1}q_jD_j.$ The supersript $s$ denotes that the operator acts on single party Hilbert space. For $j\rightarrow i$ jump, we have $E_{i\leftarrow j}^{s}=K_{i\leftarrow j}^{\dagger}K_{i\leftarrow j}=\gamma_{i\leftarrow j}D_j.$ Let us define a set:
\begin{equation}
    \mathcal{I}=\{  0 \} \cup \{  (i\leftarrow j) : 0\leq i<j\leq d-1\}
\end{equation}
For every $\alpha \in \mathcal{I}$ we have Kraus operator $K_{\alpha}.$ We define $E_{\alpha}^{s}=K_{\alpha}^{\dagger}K_{\alpha}.$ Let $T$ be defined as:
\begin{equation}
    E_{\alpha}^{s}=\sum_{i=0}^{d-1}T_{\alpha i}D_i.
\end{equation}
Let us assume that every excited level has non-zero probability of decay ($\eta_j>0$ for every $j>0$). Then, $T$ has full column rank, and hence its smallest singular
value is strictly positive,
\begin{equation}
    \sigma_{\min}(T)>0.
\end{equation}
For convenience, we define
\begin{equation}
    \mu:=\sigma_{\min}(T)^2>0.
    \label{eq::smmudef}
\end{equation}

Let $\beta=(\alpha_1,\cdots \alpha_m).$ We have:
\begin{equation}
    E_{\beta}=E_{\alpha_1}^{s}\otimes \cdots \otimes E_{\alpha_m}^{s}.
\end{equation}
Let $z= (z_1,\cdots , z_m),$ where each $z_i\in\{ 0,1,\cdots ,d-1 \}$ and :
\begin{equation}
    D_z=D_{z_1}\otimes \cdots \otimes D_{z_m}=\ket{z}\bra{z}.
\end{equation}
Similarly to the single operator case:
\begin{align}
    &E_{\beta}=E_{\alpha_1}^{s}\otimes \cdots \otimes E_{\alpha_m}^{s}\notag\\
    &=\left(\sum_{z_1=0}^{d-1}T_{\alpha_1 z_1}D_{z_1}\right)\otimes \cdots \otimes \left(\sum_{z_m=0}^{d-1}T_{\alpha_m z_m}D_{z_m}\right)=\sum_{z}M_{\beta z}D_z.\label{eq::smTwithM}
\end{align}
    
In Eq. \eqref{eq::smTwithM} we have $M_{\beta z}=\prod_{i=1}^{m}T_{\alpha_iz_i}.$ Consequently, $M=T^{\otimes m}.$ Since the singular values of a tensor product are all pairwise
products of the singular values of its factors, we have
\begin{equation}
\begin{split}
    \sigma_{\min}(M)^2
    &=
    \sigma_{\min}\left(T^{\otimes m}\right)^2\\
    &=
    \sigma_{\min}(T)^{2m}=
    \mu^m.
\end{split}
\label{eq::smMsingular}
\end{equation}

Let us define $t_z=\mathrm{Tr}(\mathcal{P}D_z\mathcal{P})$ and 
\begin{equation}
    Y_z=\mathcal{P}D_z\mathcal{P}-\frac{t_z}{k}\mathcal{P}.
    \label{eq::smyz}
\end{equation}
We have:
\begin{equation}
    a_{\beta}=\frac{1}{k}\mathrm{Tr}X_{\beta}=\sum_zM_{\beta z}\frac{t_z}{k}. 
\end{equation}
We obtain
\begin{equation}
    X_{\beta}-a_\beta \mathcal{P}=\sum_z M_{\beta z} Y_z.
\end{equation}
Inequality by the smallest singular value of $M$ gives:
\begin{equation}
    \sum_{\beta}\Vert X_{\beta}-a_{\beta}\mathcal{P}\Vert_2^{2}\geq \mu^{m}\sum_z\Vert Y_z\Vert_2^{2}.
\end{equation}
Let us estimate $\Vert Y_z\Vert_2^{2}.$ Expanding Eq. \eqref{eq::smyz} and noting that $Y_z$ is hermitian we obtain:
\begin{equation}
    \Vert Y_z\Vert_2^{2}=\mathrm{Tr}[(\mathcal{P}D_z\mathcal{P}-\frac{t_z}{k}\mathcal{P})^{2}]=\mathrm{Tr}[A_z^{2}-\frac{2t_z}{k}A_z+\frac{t_z^{2}}{k^2}\mathcal{P}],
\end{equation}
where $A_z=\mathcal{P}D_z\mathcal{P}.$ Using definition of $t_z$ and $\mathrm{Tr}\mathcal{P}=k$ and the fact that $A_z$ is at most rank one, we obtain:
\begin{equation}
    \Vert Y_z\Vert_2^{2}=(1-\frac{1}{k})t_z^{2}.
\end{equation}
Summing over all $z$ we obtain:
\begin{equation}
    \sum_z \Vert Y_z\Vert_2^{2}=(1-\frac{1}{k})\sum_zt_z^{2}.
    \label{eq::smsumyz}
\end{equation}
We further obtain:
\begin{equation}
    \sum_z t_z=\mathrm{Tr}(\mathcal{P}\sum_z D_z \mathcal{P})=\mathrm{Tr}\mathcal{P}=k,
    \label{eq::smtzk}
\end{equation}
where we used the fact that $\sum_z D_z=\openone.$ Recall that $z\in (z_1,\cdots , z_m).$ Each $z_i$ can take $d$ possible values. Therefore, there is $d^{m}$ possible choices of $z.$ Using Cauchy-Schwarz:
\begin{equation}
    (\sum_zt_z)^{2}\leq (\sum_zt_z^{2})(\sum_z 1)=d^{m}\sum_zt_z^{2}.
    \label{eq::smtd}
\end{equation}
Joining Eqs. \eqref{eq::smtzk} and \eqref{eq::smtd} we obtain:
\begin{equation}
    \sum_z t_z^{2}\geq \frac{k^{2}}{d^{m}}.
\end{equation}
Substituting it into Eq. \eqref{eq::smsumyz} we obtain:
\begin{equation}
    \sum_z \Vert Y_z\Vert_2^{2} \geq \frac{k(k-1)}{d^{m}}
\end{equation}

Combining the previous bounds, we obtain
\begin{equation}
\begin{split}
    1-
    \sum_{\boldsymbol\beta}
    \left[
        \det\nolimits_{\mathcal P}
        (\mathcal P E_{\boldsymbol\beta}\mathcal P)
    \right]^{1/k}
    &\geq
    \frac{1}{2k}
    \sum_{\boldsymbol\beta}
    \left\|
        X_{\boldsymbol\beta}
        -
        a_{\boldsymbol\beta}\mathcal P
    \right\|_2^2\\
    &\geq
    \frac{\mu^m}{2k}
    \sum_{\mathbf z}\|Y_{\mathbf z}\|_2^2\\
    &\geq
    \frac{\mu^m}{2k}
    \frac{k(k-1)}{d^m}\\
    &=
    \frac{k-1}{2}
    \left(\frac{\mu}{d}\right)^m.
\end{split}
\end{equation}
Finally, for $k=k_m$, we have
\begin{equation}
    \kappa
    \leq
    1-\frac{k_m-1}{2}
    \left(\frac{\mu}{d}\right)^m.
\end{equation}

As before, we assume that the entanglement of formation does
not vanish asymptotically. Using Lemma
\ref{le::epsilondistance}, we obtain the following lower bound
on the distance to the set of separable states:
\begin{equation}
    \varepsilon_{n,m}
    \geq
    \frac{E_0^2}
    {8m^2(\log_2 d)^2}.
\end{equation}

On the other hand, Eq. \eqref{eq::separabilityk}, applied to the block channel
$\Lambda_{\mathrm{AD}}^{\otimes m}$, gives
\begin{equation}
    \varepsilon_{n,m}
    \leq
    2
    \kappa^{n/[2(k_m-1)]}
    \left(\sqrt{k_m}-1\right),
\end{equation}
where $k_m$ is the Schmidt-rank parameter associated with the
input state and satisfies
\begin{equation}
    k_m\leq d^m.
\end{equation}

Consequently,
\begin{align}
    \kappa^{\frac{n}{2(k_m-1)}}
    &\leq
    \left[
        1-\frac{k_m-1}{2}
        \left(\frac{\mu}{d}\right)^m
    \right]^{\frac{n}{2(k_m-1)}}
    \\
    &\leq
    \exp\left[
        -\frac{n}{2(k_m-1)}
        \frac{k_m-1}{2}
        \left(\frac{\mu}{d}\right)^m
    \right]
    \\
    &=
    \exp\left[
        -\frac{n}{4}
        \left(\frac{\mu}{d}\right)^m
    \right]
    \\
    &=
    \exp\left[
        -\frac{n\mu^m}{4d^m}
    \right],
\end{align}
where we used the inequality $1-x\leq e^{-x}$.

It follows that
\begin{equation}
    \varepsilon_{n,m}
    \leq
    2\left(\sqrt{k_m}-1\right)
    \exp\left[
        -\frac{n\mu^m}{4d^m}
    \right].
\end{equation}
Using
\begin{equation}
    \sqrt{k_m}-1
    \leq
    \sqrt{k_m}
    \leq
    d^{m/2},
\end{equation}
we obtain
\begin{equation}
    \varepsilon_{n,m}
    \leq
    2d^{m/2}
    \exp\left[
        -\frac{n\mu^m}{4d^m}
    \right].
\end{equation}

Combining the lower and upper bounds on
$\varepsilon_{n,m}$ gives
\begin{equation}
    \frac{E_0^2}
    {8m^2(\log_2 d)^2}
    \leq
    2d^{m/2}
    \exp\left[
        -\frac{n\mu^m}{4d^m}
    \right].
\end{equation}
Equivalently,
\begin{equation}
    \exp\left[
        -\frac{n\mu^m}{4d^m}
    \right]
    \geq
    \frac{E_0^2}
    {16m^2(\log_2 d)^2d^{m/2}}.
\end{equation}

Taking the natural logarithm yields
\begin{equation}
    \frac{n\mu^m}{4d^m}
    \leq
    \frac{m}{2}\ln d
    +
    2\ln m
    +
    C_0,
\end{equation}
where
\begin{equation}
    C_0
    :=
    \ln\left[
        \frac{16(\log_2 d)^2}{E_0^2}
    \right]
\end{equation}
is independent of $n$ and $m$. Therefore,
\begin{equation}
    n
    \leq
    4\left(\frac{d}{\mu}\right)^m
    \left[
        \frac{m}{2}\ln d
        +
        2\ln m
        +
        C_0
    \right].
\end{equation}

For all sufficiently large $m$, the expression in square
brackets is upper bounded by a constant times $m$. Hence,
there exists a constant $C_1>0$, independent of $n$ and $m$,
such that
\begin{equation}
    n
    \leq
    C_1m
    \left(\frac{d}{\mu}\right)^m.
\end{equation}

Defining
\begin{equation}
    a
    :=
    \ln\left(\frac{d}{\mu}\right)>0,
\end{equation}
we can write
\begin{equation}
    n
    \leq
    C_1m e^{am}.
\end{equation}
Multiplying both sides by $a/C_1$, we obtain
\begin{equation}
    \frac{a}{C_1}n
    \leq
    (am)e^{am}.
\end{equation}
Since the function $x\mapsto xe^x$ is increasing for $x>0$,
the principal branch of the Lambert function gives
\begin{equation}
    m
    \geq
    \frac{1}{a}
    W_0\left(
        \frac{a}{C_1}n
    \right).
\end{equation}

Using the asymptotic expansion\cite{CorlessLambert}
\begin{equation}
    W_0(x)
    =
    \ln x-\ln\ln x
    +
    O\left(
        \frac{\ln\ln x}{\ln x}
    \right),
\end{equation}
we finally arrive at
\begin{equation}
    m
    \geq
    \frac{\ln n-\ln\ln n}
    {\ln(d/\mu)}
    -
    O(1).
\end{equation}
Therefore, preserving a nonvanishing amount of end-to-end
entanglement requires
\begin{equation}
    m=\Omega(\ln n).
\end{equation}

\begin{lem}{\label{le::epsilondistance}}
    Let $\rho_{n,m(n)}^{AB}$ denote the final state produced by a protocol
using $m$ parallel channel uses per link and $n$ sequential
links. We define its trace distance from the set of separable
states as
\begin{equation}
    \varepsilon_{n,m}
    :=
    \frac{1}{2}
    \min_{\sigma^{AB}\in\mathcal{S}}
    \left\|
        \rho_{n,m(n)}^{AB}-\sigma^{AB}
    \right\|_1 .
\end{equation}
The factor $1/2$ is introduced so that
$\varepsilon_{n,m}\in[0,1]$. There is the following implication:
\begin{equation}
    \liminf_{n\rightarrow\infty}
    E_f\!\left(\rho_{n,m(n)}^{AB}\right)
    >0 \Rightarrow \varepsilon_{n,m}
    \geq
    \frac{E_0^2}
    {8m^2(\log_2d)^2},
\end{equation}
where $E_0$ is defined by:
\begin{equation}
    E_f\!\left(\rho_{n,m(n)}^{AB}\right)
    \geq E_0
\end{equation}
for all $n\geq n_0$.
\end{lem}

We first convert this assumption into a lower bound on
$\varepsilon_{n,m}$. Let $\sigma_{n,m}\in\mathcal{S}$ be a
separable state attaining the minimum in the definition of
$\varepsilon_{n,m}$, and define
\begin{equation}
    \delta_{n,m}
    :=
    \sqrt{
        \varepsilon_{n,m}
        \left(2-\varepsilon_{n,m}\right)
    }.
\end{equation}
By the uniform continuity bound for the entanglement of
formation \cite{Winterefbound},
\begin{equation}
    \left|
        E_f(\rho_{n,m})-E_f(\sigma_{n,m})
    \right|
    \leq
    \delta_{n,m}\log_2 D_m
    +
    \left(1+\delta_{n,m}\right)
    h_2\!\left(
        \frac{\delta_{n,m}}{1+\delta_{n,m}}
    \right),\notag
\end{equation}
where $D_m$ is the smaller local dimension and $h_2$ denotes
the binary entropy. Since $\sigma_{n,m}$ is separable,
$E_f(\sigma_{n,m})=0$. Moreover, $\operatorname{dim}\mathcal{H_A}=d^{m}$ and consequently
$D_m\leq d^m$. Hence
\begin{equation}
    E_f(\rho_{n,m})
    \leq
    \delta_{n,m}m\log_2 d
    +
    \left(1+\delta_{n,m}\right)
    h_2\!\left(
        \frac{\delta_{n,m}}{1+\delta_{n,m}}
    \right).
\end{equation}
Using
\begin{equation}
    h_2(x)\leq 2\sqrt{x(1-x)},
\end{equation}
we further obtain
\begin{equation}
    E_f(\rho_{n,m})
    \leq
    \delta_{n,m}m\log_2 d
    +
    2\sqrt{\delta_{n,m}}.
    \label{eq::smef}
\end{equation}

If $m(n)$ remained bounded along an infinite subsequence, the
upper bound from Theorem \ref{thm:exponential} would imply
$\varepsilon_{n,m(n)}\rightarrow0$ and consequently
$E_f(\rho_{n,m(n)})\rightarrow0$. Thus, the assumption of
nonvanishing end-to-end entanglement implies that
$m(n)\rightarrow\infty$.

Since $m\rightarrow \infty,$ for all $E_0$ we can choose $m_0$ such that for all $m>m_0$, the following inequality holds:
\begin{equation}
    m\log_2d\geq\frac{8}{E_0}.
    \label{eq::sme08}
\end{equation}
Since $E_0$ can be arbitrary, we can choose it such that:
\begin{equation}
    \liminf_{n\rightarrow\infty}
    E_f\!\left(\rho_{n,m(n)}^{AB}\right)
    >0 \Rightarrow E_f(\rho_{n,m(n)}^{AB})>E_0
    \label{eq::smlime0}
\end{equation}
for all sufficiently large $n.$

 Assume that the following inequality holds: $\delta_{n,m}<\frac{E_0}{2m\log_2d}.$ We obtain:
    \begin{equation}
         \delta_{n,m}m\log_2 d + 2\sqrt{\delta_{n,m}}< \frac{E_0}{2}+2\sqrt{\frac{E_0}{2m\log_2 d}}\leq \frac{E_0}{2}+\frac{E_0}{2},
    \end{equation}
    where in the last step we used Eq. \eqref{eq::sme08}.
    But combination of Eqs. \eqref{eq::smef} and \eqref{eq::smlime0} states that $E_0<E_f(\rho_{n,m})\leq \delta_{n,m}m\log_2 d + 2\sqrt{\delta_{n,m}}.$ This contradiction implies that:
\begin{equation}
    \delta_{n,m}
    \geq
    \frac{E_0}{2m\log_2d}.
    \label{eq::smco}
\end{equation}

Since
\begin{equation}
    \delta_{n,m}^2
    =
    \varepsilon_{n,m}
    \left(2-\varepsilon_{n,m}\right)
    \leq
    2\varepsilon_{n,m},
\end{equation}
we obtain the dimension-dependent lower bound
\begin{equation}
    \varepsilon_{n,m}
    \geq
    \frac{E_0^2}
    {8m^2(\log_2d)^2}
\end{equation}
for all sufficiently large $n$.

\subsection*{Achievability of the bound for the depolarizing channel}

Consider the $d$-dimensional depolarizing channel
\begin{equation}
    \Lambda_p(\rho)
    =
    (1-p)\rho
    +
    p\frac{\openone}{d}.
\end{equation}
Using the Weyl twirling identity\cite{KhatriPQCT}, the channel can equivalently
be written as
\begin{equation}
    \Lambda_p(\rho)
    =
    \left(1-p+\frac{p}{d^2}\right)\rho
    +
    \frac{p}{d^2}
    \sum_{(a,b)\neq(0,0)}
    W_{ab}\rho W_{ab}^{\dagger}.
    \label{eq:depolarizing-weyl}
\end{equation}
The probability of each non-identity Weyl error is $\frac{p}{d^{2}},$ since there is $d^{2}-1$ possible errors, probability that there will be an error on some qudit is 
\begin{equation}
    q
    =
    p\left(1-\frac{1}{d^2}\right).
    \label{eq:weyl-error-probability}
\end{equation}

Suppose that one elementary link consists of $m$ parallel and
independent uses of $\Lambda_p$. Let $X_m$ denote the number of
physical qudits affected by nonidentity Weyl errors in this block.
It follows from Eq.~\eqref{eq:weyl-error-probability} that
\begin{equation}
    X_m\sim\operatorname{Bin}(m,q),
\end{equation}
and hence
\begin{equation}
    \Pr(X_m=r)
    =
    \binom{m}{r}q^r(1-q)^{m-r}.
    \label{eq:binomial-distribution}
\end{equation}

Assume that there exists a family of quantum codes
$\{\mathcal C_m\}$, encoding at least one logical qudit, whose
minimum distance $\Delta_m$ satisfies
\begin{equation}
    \Delta_m\geq \delta m
\end{equation}
for some constant $\delta>0$ and all sufficiently large $m$.
A code with distance $\Delta_m$ corrects every error supported
on at most
\begin{equation}
    t_m
    =
    \left\lfloor\frac{\Delta_m-1}{2}\right\rfloor
\end{equation}
physical qudits\cite{Nielsen_Chuang_2010}.

Choose a constant $\tau$ such that
\begin{equation}
    q<\tau<\frac{\delta}{2}.
    \label{eq:tau-condition}
\end{equation}
The crucial point is that the code distance grows linearly with
the number of physical qudits. $\Delta_m \geq \delta m$. Therefore, for every
fixed $\tau<\delta/2$ and all sufficiently large $m$,
\begin{equation}
    \lfloor \tau m\rfloor
    \leq
    \left\lfloor\frac{\Delta_m-1}{2}\right\rfloor.
\end{equation}
Hence, every Weyl-error pattern of weight at most
$\lfloor\tau m\rfloor$ is correctable. Therefore, the failure
probability of the recovery operation on one elementary link
satisfies
\begin{equation}
    \varepsilon_m
    \leq
    \Pr(X_m>\tau m).
\end{equation}
Since $\tau>q$, the Chernoff-Hoeffding bound gives~\cite{Hoeffdingbound}
\begin{equation}
    \varepsilon_m
    \leq
    \exp\left[-mD(\tau\Vert q)\right],
    \label{eq:single-link-failure}
\end{equation}
where
\begin{equation}
    D(\tau\Vert q)
    =
    \tau\ln\frac{\tau}{q}
    +
    (1-\tau)\ln\frac{1-\tau}{1-q}
\end{equation}
is the binary relative entropy. We denote
\begin{equation}
    \alpha:=D(\tau\Vert q)>0.
\end{equation}

After each elementary link, the corresponding repeater station
applies the recovery operation and re-encodes the logical system
before transmission through the next link. The probability $\varepsilon_{n,m}$ that
at least one of the $n$ recovery operations fails is bounded,
by the union bound, as
\begin{equation}
    \varepsilon_{n,m}
    \leq
    n\varepsilon_m
    \leq
    ne^{-\alpha m}.
    \label{eq:network-failure}
\end{equation}

For any constant $\eta>0$, choose
\begin{equation}
    m(n)
    =
    \left\lceil
    \frac{(1+\eta)\ln n}{\alpha}
    \right\rceil.
    \label{eq:achievability-block-size}
\end{equation}
It then follows that
\begin{equation}
    \varepsilon_{n,m(n)}
    \leq
    n^{-\eta},
\end{equation}
and consequently
\begin{equation}
    \lim_{n\rightarrow\infty}
    \varepsilon_{n,m(n)}
    =
    0.
\end{equation}

To relate this result to entanglement distribution, let the
logical input be one half of a maximally entangled logical state
$\Phi_L$. If all error patterns are correctable, the state
transmitted through the entire network is exactly $\Phi_L$.
The final state can therefore be written as
\begin{equation}
    \rho_{n,m}
    =
    (1-\varepsilon'_{n,m})\Phi_L
    +
    \varepsilon'_{n,m}\sigma_{n,m},
\end{equation}
where $\varepsilon'_{n,m}\leq\varepsilon_{n,m}$ and
$\sigma_{n,m}$ is some quantum state. Consequently,
\begin{equation}
    \frac{1}{2}
    \left\|
        \rho_{n,m(n)}-\Phi_L
    \right\|_1
    \leq
    \varepsilon_{n,m(n)}
    \longrightarrow 0.
\end{equation}
By the continuity of the entanglement of formation\cite{Winterefbound},
\begin{equation}
    \lim_{n\rightarrow\infty}
    E_f\left(\rho_{n,m(n)}\right)
    =
    E_f(\Phi_L)
    =
    \log_2 d_L
    >
    0,
\end{equation}
where $d_L\geq2$ is the logical dimension.

We have thus shown that
\begin{equation}
    m=O(\ln n)
\end{equation}
parallel channel uses per elementary link are sufficient.
Together with the lower bound of Theorem \ref{th::scaling}, this establishes the optimal scaling
\begin{equation}
    m=\Theta(\ln n)
\end{equation}
for depolarizing noise satisfying the condition
$q<\delta/2$ for the chosen family of codes.

\end{document}